\documentclass[thm-restate,a4paper,UKenglish,cleveref, autoref, thm-restate]{lipics-v2021}

\usepackage{amsmath, amssymb, amsfonts, amsthm}
\usepackage[T1]{fontenc}
\usepackage{thmtools}
\usepackage{algorithmicx}
\usepackage{algorithm}
\usepackage{algpseudocode}
\usepackage{graphicx}
\usepackage{cancel}
\usepackage{textcomp}
\usepackage{booktabs}
\usepackage{tikz}
\usepackage{bm}
\usepackage{standalone}
\usepackage{accents}
\usepackage{lineno}

\title{Construction of Sparse Suffix Trees and LCE~Indexes in~Optimal~Time and Space} 

\author{Dmitry Kosolobov}{Ural Federal University, Ekaterinburg, Russia}{dkosolobov@mail.ru}{0000-0002-2909-2952}{}

\author{Nikita Sivukhin}{Ural Federal University, Ekaterinburg, Russia}{sivukhin.nikita@yandex.ru}{0000-0003-4995-6954}{}

\funding{Supported by the Russian Science Foundation (RSF), project 20-71-00050.}

\authorrunning{D. Kosolobov and N. Sivukhin} 

\Copyright{Dmitry Kosolobov and Nikita Sivukhin} 



\category{} 
\newcommand\ArxivVersion{1}
\ifdefined\ArxivVersion
\else
\relatedversion{The full version of the paper is available at arXiv} 
\relatedversiondetails{Full version}{https://arxiv.org/abs/2105.03782}
\fi

\nolinenumbers 
\ifdefined\ArxivVersion
  \hideLIPIcs
\fi

\EventEditors{Shunsuke Inenaga and Simon J. Puglisi}
\EventNoEds{2}
\EventLongTitle{35th Annual Symposium on Combinatorial Pattern Matching (CPM 2024)}
\EventShortTitle{CPM 2024}
\EventAcronym{CPM}
\EventYear{2024}
\EventDate{June 25--27, 2024}
\EventLocation{Fukuoka, Japan}
\EventLogo{}
\SeriesVolume{296}
\ArticleNo{24}

\newcommand\Oh{\mathcal{O}}
\newcommand\lce{\mathop{\mathsf{lce}}}
\newcommand\lbit{\mathop{\mathsf{bit}}}
\newcommand\vbit{\mathop{\mathsf{vbit}}}
\newcommand{\appendixv}[2]{%
\ifdefined\ArxivVersion%
Appendix~\ref{#1}%
\else%
Appendix~#2 in the full version~\cite{self}%
\fi%
}

\begin{document}
\maketitle

\begin{abstract}
The notions of synchronizing and partitioning sets are recently introduced variants of locally consistent parsings with a great potential in problem-solving. In this paper we propose a deterministic algorithm that constructs for a given readonly string of length $n$ over the alphabet $\{0,1,\ldots,n^{\Oh(1)}\}$ a variant of a $\tau$-partitioning set with size $\Oh(b)$ and $\tau = \frac{n}{b}$ using $\Oh(b)$ space and $\Oh(\frac{1}{\epsilon}n)$ time provided $b \ge n^\epsilon$, for $\epsilon > 0$. As a corollary, for $b \ge n^\epsilon$ and constant $\epsilon > 0$, we obtain linear time construction algorithms with $\Oh(b)$ space on top of the string for two major small-space indexes: a~sparse suffix tree, which is a compacted trie built on $b$ chosen suffixes of the string, and a \emph{longest common extension} (LCE) index, which occupies $\Oh(b)$ space and allows us to compute the longest common prefix for any pair of substrings in $\Oh(n/b)$ time. For both, the $\Oh(b)$ construction storage is asymptotically optimal since the tree itself takes $\Oh(b)$ space and any LCE index with $\Oh(n/b)$ query time must occupy at least $\Oh(b)$ space by a known trade-off (at least for $b \ge \Omega(n / \log n)$). In case of arbitrary~$b \ge \Omega(\log^2 n)$, we present construction algorithms for the partitioning set, sparse suffix tree, and LCE index with $\Oh(n\log_b n)$ running time and $\Oh(b)$ space, thus also improving the state of the art.
\end{abstract}

\algtext*{EndIf}
\algtext*{EndWhile}
\algtext*{EndFor}

\newpage

\section{Introduction}

Indexing data structures traditionally play a central role in algorithms on strings and in information retrieval.  Due to constantly growing volumes of data in applications, the attention of researchers in the last decades was naturally attracted to small-space indexes. 
In this paper we study two closely related small-space indexing data structures: a sparse suffix tree and a longest common extension (LCE) index. We investigate them in the general framework of (deterministic) locally consistent parsings that was developed by Cole and Vishkin~\cite{ColeVishkin}, Je{\.z}~\cite{Jez,Jez2,Jez3,Jez4}, and others~\cite{AlstrupBrodalRauhe,FischerIKoppl,GanczorzGawrychowskiJezKociumaka,GawrychowskiEtAl,GoldbergPlotkinShannon,MelhornSundarUhrig,NishimotoEtAl,SahinalpVishkin}  (the list is not exhausting) 
and was recently revitalized in the works of Birenzwige et al.~\cite{BirenzwigeEtAl} and~Kempa and Kociumaka~\cite{KempaKociumaka} where two new potent concepts of partitioning and synchronizing sets were introduced.

The sparse suffix tree (\emph{SST}) for a given set of $b$ suffixes of a string is a compacted trie built on these suffixes. 
It can be viewed as the suffix tree from which all suffixes not from the set were removed (details follow). 
The tree takes $\Oh(b)$ space on top of the input string and can be easily constructed in $\Oh(n)$ time from the suffix tree, where $n$ is the length of the string. One can build the suffix tree in $\Oh(n)$ time~\cite{Farach} provided the letters of the string are sortable in linear time. However, if at most $\Oh(b)$ space is available on top of the input, then in general there is not enough memory for the full suffix tree and the problem, thus, becomes much more difficult. The $\Oh(b)$ bound is optimal since the tree itself takes $\Oh(b)$ space. The construction problem with restricted $\Oh(b)$ space naturally arises in applications of the sparse suffix tree and the sparse suffix array (which is easy to retrieve from the tree) where we have to index data in the setting of scarce memory. As is common in algorithms on strings, it is assumed that the input string is readonly, its letters are polynomially bounded integers $\{0,1,\ldots,n^{\Oh(1)}\}$, and the space is at least polylogarithmic, i.e., $b \ge \log^{\Omega(1)} n$. We note, however, that in supposed usages the memory restrictions can often be relaxed even more to $b \ge n^\epsilon$, for constant~$\epsilon > 0$.

The $\Oh(b)$-space construction problem was posed by K{\"a}rkk{\"a}inen and Ukkonen~\cite{KarkkainenUkkonen} who showed how to solve it in linear time for the case of evenly spaced $b$ suffixes. In a series of works~\cite{BilleEtAl3,BirenzwigeEtAl,FischerIKoppl,GawrychowskiKociumaka,IEtAl,KarkkainenSandersBurkhardt}, the problem was settled for the case of randomized algorithms: an optimal linear $\Oh(b)$-space Monte Carlo construction algorithm for the sparse suffix tree was proposed by Gawrychowski and Kociumaka~\cite{GawrychowskiKociumaka} and an optimal linear $\Oh(b)$-space Las-Vegas algorithm was described by Birenzwige et al.~\cite{BirenzwigeEtAl}. The latter authors also presented the best up-to-date deterministic solution that builds the sparse suffix tree within $\Oh(b)$ space in $\Oh(n \log\frac{n}b)$ time~\cite{BirenzwigeEtAl} 
($\log$ is in base~2 unless explicitly stated otherwise). 
All these solutions assume (as we do too) that the input string is readonly and its alphabet is $\{0,1,\ldots,n^{\Oh(1)}\}$; the case of rewritable inputs is apparently very different, as was shown by Prezza~\cite{Prezza}.

The LCE index, crucial in string algorithm applications, preprocesses a readonly input string of length $n$ so that one can answer queries $\lce(p,q)$, for any positions $p$ and $q$, computing the length of the longest common prefix of the suffixes starting at $p$ and $q$. The now classical result of Harel and Tarjan states that the LCE queries can be answered in $\Oh(1)$ time provided $\Oh(n)$ space is used~\cite{HarelTarjan}. In~\cite{BilleEtAl2} Bille et al.~presented an LCE index that, for any given user-defined parameter $b$, occupies $\Oh(b)$ space on top of the input string and answers queries in $\Oh(\frac{n}{b})$ time. In~\cite{Kosolobov} it was proved that this time-space trade-off is optimal provided $b \ge \Omega(n / \log n)$ (it is conjectured that the same trade-off lower bound holds for a much broader range of values $b$; a weaker trade-off appears in~\cite{BilleEtAl4,BrodalDavoodiRao}). In view of these lower bounds, it is therefore natural to ask how fast one can construct, for any parameter $b$, an LCE index that can answer queries in $\Oh(\frac{n}{b})$ time using $\Oh(b)$ space on top of the input. The space $\Oh(b)$ is optimal for this query time and the construction algorithm should not exceed it. The issue with the data structure of~\cite{BilleEtAl2} is that its construction time is unacceptably slow, which motivated a series of works trying to solve this problem. As in the case of sparse suffix trees, the problem was completely settled in the randomized setting: an optimal linear $\Oh(b)$-space Monte Carlo construction algorithm for an LCE index with $\Oh(\frac{n}{b})$-time queries was presented by Gawrychowski and Kociumaka~\cite{GawrychowskiKociumaka} and a Las-Vegas construction with the same time and space was proposed by Birenzwige et al.~\cite{BirenzwigeEtAl} provided $b \ge \Omega(\log^2 n)$. The best deterministic solution is also presented in~\cite{BirenzwigeEtAl} and runs in $\Oh(n\log\frac{n}{b})$ time answering queries in slightly worse time $\Oh(\frac{n}{b}\sqrt{\log^* n})$ provided $b \ge \Omega(\log n)$ (the previous best solution was from~\cite{TanimuraEtAl} and it runs in $\Oh(n\cdot\frac{n}{b})$ time but, for some exotic parameters $b$, has slightly better query time).
The input string is readonly in all these solutions and the alphabet is $\{0,1,\ldots,n^{\Oh(1)}\}$.

For a broad range of values $b$, we settle both construction problems, for sparse suffix trees and LCE indexes, in $\Oh(b)$ space in the deterministic case. Specifically, given a readonly string of length $n$ over the alphabet $\{0,1,\ldots,n^{\Oh(1)}\}$, we present two algorithms: one that constructs the sparse suffix tree, for any user-defined set of $b$ suffixes such that $b \ge \Omega(\log^2 n)$, in $\Oh(n \log_b n)$ time using $\Oh(b)$ space on top of the input; and another that constructs an LCE index with $\Oh(\frac{n}{b})$-time queries, for any parameter $b$ such that $b \ge \Omega(\log^2 n)$, in $\Oh(n \log_b n)$ time using $\Oh(b)$ space on top of the input. This gives us optimal $\Oh(b)$-space solutions with $\Oh(\frac{1}{\epsilon} n) = \Oh(n)$ time when $b \ge n^\epsilon$, for constant $\epsilon > 0$, which arguably includes most interesting cases. As can be seen in Table~\ref{tbl:results}, our result beats the previous best solution in virtually all settings since $n \log_b n = o(n\log\frac{n}{b})$, 
for $b = o(n)$.

\begin{table}[t]
\caption{LCE indexes deterministically constructible in $\Oh(b)$ space on a readonly input, for $b \ge \Omega(\log^2 n)$.}
\centering
\begin{tabular}{r|c|c|c}\hline
Algorithm & Tanimura et al.~\cite{TanimuraEtAl} & Birenzwige et al.~\cite{BirenzwigeEtAl} & Theorem~\ref{thm:sparse-lce} \\\hline
Query time & $\Oh(\frac{n}{b} \log\min\{b,\frac{n}{b}\})$ & $\Oh(\frac{n}{b} \sqrt{\log^* n})$ & $\bm{\Oh(\frac{n}{b})}$ \\\hline
Construction in $\Oh(b)$ space & $\Oh(n\cdot \frac{n}{b})$ & $\Oh(n\log\frac{n}{b})$ & $\bm{{\Oh(n\log_b n)}}$ \\\hline
\end{tabular}
\label{tbl:results}
\end{table}

In order to achieve these results, we develop a new algorithm that, for any given parameter $b \ge \Omega(\log^2 n)$, constructs a so-called $\tau$-partitioning set of size $\Oh(b)$ with $\tau = \frac{n}{b}$. This result is of independent interest.

We note that there is another natural model where the input string is packed in memory in such a way that one can read in $\Oh(1)$ time any $\Theta(\log_\sigma n)$ consecutive letters of the input packed into one $\Theta(\log n)$-bit machine word, where $\{0,1,\ldots, \sigma{-}1\}$ is the input alphabet. In this case the $\Oh(n)$ construction time is not necessarily optimal for the sparse suffix tree and the LCE index and one might expect to have $\Oh(n / \log_\sigma n)$ time. As was shown by Kempa and Kociumaka~\cite{KempaKociumaka}, this is indeed possible for LCE indexes in $\Oh(n / \log_\sigma n)$ space. It remains open whether one can improve our results for the $\Oh(b)$-space construction in this setting; 
note that the lower bound of~\cite{Kosolobov} does not apply here due to its assumption of single-letter input memory cells.

\subparagraph{Techniques.}
The core of our solution is a version of locally consistent parsing developed by Birenzwige et al.~\cite{BirenzwigeEtAl}, the so-called $\tau$-partitioning sets (unfortunately, we could not adapt the more neat $\tau$-synchronizing sets from~\cite{KempaKociumaka} for the deterministic case).
It was shown by Birenzwige et al.~that the $\Oh(b)$-space construction of a sparse suffix tree or an LCE index can be performed in linear time provided a $\tau$-partitioning set of size $\Oh(b)$ with $\tau = \frac{n}{b}$ is given. We define a variant of $\tau$-partitioning sets and, for completeness, repeat the argument of Birenzwige et al.~with minor adaptations to our case. The main bulk of the text is devoted to the description of an $\Oh(b)$-space algorithm that builds a (variant of) $\tau$-partitioning set of size $\Oh(b)$ with $\tau = \frac{n}{b}$ in $\Oh(n\log_b n)$ time provided $b \ge \Omega(\log^2 n)$, which is the main result of the present paper. In comparison Birenzwige et al.'s algorithm for their $\tau$-partitioning sets runs in $\Oh(n)$ \emph{expected} time (so that it is a Las Vegas construction) and $\Oh(b)$ space; their deterministic algorithm takes $\Oh(n\log\tau)$ time but the resulting set is only $\tau\log^* n$-partitioning. 
Concepts very similar to partitioning sets appeared also in~\cite{RobertsEtAl,SchleimerWilkersonAiken}.

Our solution combines two well-known approaches to deterministic locally consistent parsings: the \emph{deterministic coin tossing} introduced by Cole and Vishkin~\cite{ColeVishkin} and developed in~\cite{AlstrupBrodalRauhe,FischerIKoppl,GanczorzGawrychowskiJezKociumaka,GawrychowskiEtAl,GoldbergPlotkinShannon,MelhornSundarUhrig,NishimotoEtAl,SahinalpVishkin}, and the \emph{recompression} invented by Je{\.z}~\cite{Jez4} and studied in~\cite{I,Jez,Jez2,Jez3}. The high level idea is first to use Cole and Vishkin's technique that constructs a $\tau$-partitioning set of size $\Oh(b\log^* n)$ where $\tau =\frac{n}{b}$ (in our algorithm the size is actually $\Oh(b\log\log\log n)$ since we use a ``truncated'' version of Cole and Vishkin's bit reductions); second, instead of storing the set explicitly, which is impossible in $\Oh(b)$ space, we construct a string $R$ of length $\Oh(b\log^* n)$ in which every letter corresponds to a position of the set and occupies $o(\log\log n)$ bits so that $R$ takes $o(b\log^* n \log\log n)$ bits in total and, thus, can be stored into $\Oh(b)$ machine words of size $\Oh(\log n)$ bits; third, Je{\.z}'s recompression technique is iteratively applied to the string $R$ until $R$ is shortened to length $\Oh(b)$; finally, the first technique generating a $\tau$-partitioning set is performed again but this time we retain and store explicitly those positions that correspond to surviving letters of the string $R$. There are many hidden obstacles on this path and because of them our solution is only of purely theoretical value in its present form due to numerous internal complications in the actual scheme (in contrast, randomized results on synchronizing sets~\cite{DinklageEtAl,KempaKociumaka} seem quite practical).

The paper is organized as follows. In Section~\ref{sec:partitioning-sets} we define $\tau$-partitioning sets and show how one can use them to build an LCE index. Section~\ref{sec:vishkin-process} describes the first stage of the construction of a $\tau$-partitioning set that is based on a modification of Cole and Vishkin's technique. Section~\ref{sec:time-improvement} improves the running time of this stage from $\Oh(n \log\tau)$ to $\Oh(n\log_b \tau)$. In Section~\ref{sec:recompression} the second stage based on a modification of Je{\.z}'s recompression technique is presented. \appendixv{appx:small-tau}{C} describes separately the case of very small $\tau$.

\section{Partitioning Sets with Applications}
\label{sec:partitioning-sets}

Let us fix a readonly string $s$ of length $n$ whose letters $s[0], s[1], \ldots, s[n{-}1]$ are from a polynomially bounded alphabet $\{0,1,\ldots,n^{\Oh(1)}\}$. We use $s$ as the input in our algorithms. As is standard, the algorithms are in the word-RAM model, their space is measured in $\Theta(\log n)$-bit machine words, and each $s[i]$ occupies a separate word. We write $s[i..j]$ for the \emph{substring} $s[i]s[i{+}1]\cdots s[j]$, assuming it is empty if $i > j$; $s[i..j]$ is called a \emph{suffix} (resp., \emph{prefix}) of $s$ if $j = n - 1$ (resp., $i = 0$). For any string $t$, let $|t|$ denote its length. We say that $t$ \emph{occurs} at position $i$ in $s$ if $s[i..i{+}|t|{-}1] = t$. Denote $[i..j] = \{k \in \mathbb{Z} \colon i\le k\le j\}$, $(i..j] = [i..j]\setminus \{i\}$, $[i..j) = [i..j]\setminus\{j\}$, $(i..j) = [i..j)\cap (i..j]$. A~number $p \in [1..|t|]$ is called a \emph{period} of $t$ if $t[i] = t[i-p]$ for each $i \in [p..|t|)$. For brevity, denote $\log\log\log n$ by $\log^{(3)} n$. We assume that $n$, the length of $s$, is sufficiently large: larger than $2^{\max\{16,c\}}$, where $c$ is a constant such that $n^c$ upper-bounds the alphabet. 

Given an integer $\tau \in [4..n/2]$, a set of positions $S \subseteq [0..n)$ is called a~\mbox{\emph{$\tau$-partitioning set}} if it satisfies the following properties:
\begin{enumerate}
\item[(a)] if $s[i{-}\tau..i{+}\tau] = s[j{-}\tau..j{+}\tau]$ for $i,j \in [\tau..n{-}\tau)$, then $i \in S$ iff $j \in S$;
\item[(b)] if $s[i..i{+}\ell] = s[j..j{+}\ell]$, for $i,j \in S$ and some $\ell \ge 0$, then, for each $d \in [0..\ell{-}\tau)$, $i + d \in S$ iff $j + d \in S$;
\item[(c)] if $i,j \in S \cup \{ 0, n{-}1 \}$ with $j{-}i > \tau$ and $(i..j) \cap S = \emptyset$, then the period of $s[i .. j]$ is at most $\tau / 4$.
\end{enumerate}

Our definition is inspired by the \emph{forward synchronized \mbox{$(\tau,\tau)$-partitioning} sets} from~{\cite[Def.~3.1 and~6.1]{BirenzwigeEtAl}} but slightly differs; nevertheless, we retain the term ``partitioning'' to avoid inventing unnecessary new terms for very close concepts. In the definition, (a), (b), and (c) state, respectively, that $S$ is locally consistent, forward synchronized, and dense: the choice of positions depends only on short substrings around them, long enough equal right ``contexts'' of positions from $S$ are ``partitioned'' identically, and $S$ has a position every $\tau$ letters unless a long range with small period is encountered. In our construction of $S$ a certain converse of~(c) will also hold: whenever a substring $s[i..j]$ has a period at most $\tau / 4$, we will have $S \cap [i + \tau .. j - \tau] = \emptyset$ (see Lemma~\ref{lem:main-lemma}). This converse is not in the definition since it is unnecessary for our applications and we will use auxiliary $\tau$-partitioning sets not satisfying it. The definition also implies the following convenient property of ``monotonicity''.

\begin{lemma}
For any $\tau' \ge \tau$, every $\tau$-partitioning set is also \mbox{$\tau'$-partitioning}.\label{lem:monotone-synch}
\end{lemma}

Due to~(c), all $\tau$-partitioning sets in some strings have size at least $\Omega(n / \tau)$. In the remaining sections we devise algorithms that construct a $\tau$-partitioning set of $s$ with size $\Oh(n / \tau)$ (matching the lower bound) using $\Oh(n / \tau)$ space on top of~$s$; for technical reasons, we assume that $\Omega(\log^2 n)$ space is always available, i.e., $n / \tau \ge \Omega(\log^2 n)$, which is a rather mild restriction. Thus, we shall prove the following main theorem.

\begin{theorem}
For any string of length $n$ over an alphabet $[0..n^{\Oh(1)}]$ and any $\tau \in [4..\Oh(n / \log^2 n)]$, one can construct in $\Oh(n\log_b n)$ time and $\Oh(b)$ space on top of the string a $\tau$-partitioning set of size $\Oh(b)$, for $b = n / \tau$.
\label{thm:main-theorem}
\end{theorem}

Let us sketch how one can construct indexes with the $\tau$-partitioning set of Theorem~\ref{thm:main-theorem}.

\subparagraph{LCE index and sparse suffix tree.}
An LCE index is a data structure on $s$ that, given a pair of positions $p$ and $q$, answers the \emph{LCE query} $\lce(p,q)$ computing the length of the longest common prefix of $s[p..n{-}1]$ and $s[q..n{-}1]$. Such index can be stored in $\Oh(b)$ space on top of $s$ with $\Oh(\frac{n}b)$ query time~\cite{BilleEtAl2} and this trade-off is optimal, at least for $b \ge \Omega(\frac{n}{\log n})$~\cite{Kosolobov}.

Given $b$ suffixes $s[i_1..n{-}1], s[i_2..n{-}1],\ldots,s[i_b..n{-}1]$, their \emph{sparse suffix tree} \cite{KarkkainenUkkonen} is a compacted trie on these suffixes in which all edge labels are stored as pointers to corresponding substrings of $s$. Thus, the tree occupies $\Oh(b)$ space.

Our construction scheme for these two indexes is roughly as follows: given a \mbox{$\tau$-partitioning} set $S$ with $\tau = \frac{n}{b}$ and size $\Oh(b) = \Oh(n / \tau)$, we first build the sparse suffix tree $T$ for the suffixes $s[j..n{-}1]$ with $j \in S$, then we use it to construct an LCE index, and, using the index, build the sparse suffix tree for arbitrarily chosen $b$ suffixes. We elaborate on this scheme below; our exposition, however, is rather sketchy and some details are omitted since the scheme is essentially the same as in~\cite{BirenzwigeEtAl} and is given here mostly for completeness.


To construct the sparse suffix tree $T$ for all $s[j..n{-}1]$ with $j \in S$, we apply the following lemma. Its cumbersome formulation is motivated by its subsequent use in Section~\ref{sec:time-improvement}. In the special case when $m = n$ and $\sigma = n^{\Oh(1)}$, which is of primary interest for us now, the lemma states that $T$ can be built in $\Oh(n)$ time: this case implies that $m \log_b \sigma = \Oh(n \log_b n)$ is $\Oh(n)$ if $b > n / \log n$, and $b\log b$ is $\Oh(n)$ if $b \le n / \log n$, and, therefore, $\min\{m \log_b \sigma, b \log b\} = \Oh(n)$. The proof essentially follows arguments of~\cite{BirenzwigeEtAl} and is given in \appendixv{appx:sst-special}{A}.

\begin{restatable}{lemma}{sstSpecial}
Given an integer $\tau \ge 4$ and a read-only string $s$ of length $m$ over an alphabet $[0..\sigma)$, let $S$ be an ``almost'' $\tau$-partitioning set of size $b = \Theta(m / \tau)$: it satisfies properties (a) and (b), but not necessarily (c). The sparse suffix tree $T$ for all suffixes $s[j..m{-}1]$ with $j \in S$ can be built in $\Oh(m + \min\{m \log_b \sigma, b \log b\})$ time and $\Oh(m / \tau)$ space on top of the space required for $s$.
\label{lem:sst-special}
\end{restatable}

For our LCE index, we equip $T$ with the lowest common ancestor (LCA) data structure~\cite{HarelTarjan}, which allows us to compute $\lce(p,q)$ in $\Oh(1)$ time for $p, q \in S$, and we preprocess an array $N[0..b{-}1]$ such that $N[i] = \min\{j \ge i\tau \colon j \in S\}$ for $i \in [0..b)$, which allows us to calculate $\min\{j \ge p \colon j \in S\}$, for any $p$, in $\Oh(\tau)$ time by traversing $j_k, j_{k+1}, \ldots$ in $S$, for $j_k = N[\lfloor p / \tau \rfloor]$. In order to answer an arbitrary query $\lce(p,q)$, we first calculate $p' = \min\{j \ge p + \tau \colon j\in S\}$ and $q' = \min\{j \ge q + \tau \colon j\in S\}$ in $\Oh(\tau)$ time. If either $p' - p \le 2\tau$ or $q' - q \le 2\tau$, then by the local consistency of $S$, $s[p..n{-}1]$ and $s[q..n{-}1]$ either differ in their first $3\tau$ positions, which is checked na{\"i}vely, or $s[p..p'] = s[q..q']$ and the answer is given by $p' - p + \lce(p', q')$ using $T$. If $\min\{p' - p, q' - q\} > 2\tau$, then the strings $s[p{+}\tau..p']$ and $s[q{+}\tau..q']$ both have periods at most $\tau / 4$ due to property~(c); we compare $s[p..p{+}2\tau]$ and $s[q..q{+}2\tau]$ na{\"i}vely and, if there are no mismatches, therefore, due to periodicity, $s[p{+}\tau .. p']$ and $s[q{+}\tau .. q']$ have a common prefix of length $\ell = \min\{p' - p, q' - q\} - \tau$; hence, the problem is reduced to $\lce(p + \ell, q + \ell)$, which can be solved as described above since either $p' - (p + \ell) \le 2\tau$ or $q' - (q + \ell) \le 2\tau$. We thus have proved the following theorem.

\begin{theorem}
For any string of length $n$ over an alphabet $[0..n^{\Oh(1)}]$ and any $b \ge \Omega(\log^2 n)$, one can construct in $\Oh(n\log_b n)$ time and $\Oh(b)$ space on top of the string an LCE index that can answer LCE queries in $\Oh(n / b)$ time.\label{thm:sparse-lce}
\end{theorem}

Let us consider the construction of the SST for $b$ suffixes $s[i_1..n{-}1]$, $s[i_2..n{-}1], \ldots, s[i_b..n{-}1]$. Denote by $j_k$ the $k$th position in a given $\tau$-partitioning set $S$ of size $\Oh(b)$ with $\tau = \frac{n}{b}$ (so that $j_1 < \cdots < j_{|S|}$). For each suffix $s[i_\ell..n{-}1]$, we compute in $\Oh(\tau)$ time using the array $N$ an index $k_\ell$ such that $j_{k_\ell} = \min\{j \ge i_\ell + \tau \colon j \in S\}$. It takes $\Oh(b\tau) = \Oh(n)$ time in total. Then, we sort all strings $s[i_\ell..i_\ell{+}4\tau]$ in $\Oh(n)$ time as in the proof of Lemma~\ref{lem:sst-special} and assign to them ranks $r_\ell$ (equal strings are of equal ranks). For each $k \in [1..|S|]$, we obtain from the tree $T$ the rank $\bar{r}_k$ of $s[j_k..n{-}1]$ among the suffixes $s[j..n{-}1]$ with $j \in S$. Suppose that $j_{k_\ell} \le i_\ell + 3\tau$, for all $\ell \in [1..b]$. By property~(a), the equality $r_\ell = r_{\ell'}$, for any $\ell, \ell' \in [1..b]$, implies that $j_{k_\ell} - i_\ell = j_{k_{\ell'}} - i_{\ell'}$ when $j_{k_\ell} - i_\ell \le 3\tau$. Then, we sort the suffixes  $s[i_\ell..n{-}1]$ with $\ell \in [1..b]$ in $\Oh(b)$ time using the radix sort on the corresponding pairs $(r_\ell, \bar{r}_{j_{k_\ell}})$. The SST can be assembled from the sorted suffixes in $\Oh(b\tau) = \Oh(n)$ time using the LCE index to calculate longest common prefixes of adjacent suffixes.

The argument is more intricate when the condition $j_{k_\ell} > i_\ell + 3\tau$ does not hold.
Suppose that $j_{k_\ell} > i_\ell + 3\tau$, for some $\ell \in [1..b]$. Then, by property~(c), the minimal period of $s[i_\ell{+}\tau..j_{k_\ell}]$ is at most $\tau / 4$. Denote this period by $p_{\ell}$. We compute $p_{\ell}$ in $\Oh(\tau)$ time using a linear $\Oh(1)$-space algorithm~\cite{CrochemoreRytter2} and, then, we find the leftmost position $t_{\ell} > j_{k_\ell}$ breaking this period: $s[t_\ell] \ne s[t_\ell{-}p_\ell]$. As $j_{k_\ell}{-}p_\ell > i_\ell{+}2 \tau > j_{k_\ell-1}$, we obtain $s[j_{k_\ell}{-}\tau..j_{k_\ell}{+}\tau] \ne s[j_{k_\ell}{-}p_\ell{-}\tau..j_{k_\ell}{-}p_\ell{+}\tau]$ (since otherwise $j_{k_\ell}{-}p_\ell \in S$ by property~(a)) and, hence, $t_\ell \in (j_{k_\ell}..j_{k_\ell}{+}\tau]$. Therefore, the computation of $t_\ell$ takes $\Oh(\tau)$ time. Thus, all $p_\ell$ and $t_\ell$ can be calculated in $\Oh(b\tau) = \Oh(n)$ total time. We then sort the strings $s[t_\ell..t_\ell{+}\tau]$ in $\Oh(n)$ time and assign to them ranks $\tilde{r}_\ell$. For each suffix $s[i_\ell..n{-}1]$ with $\ell \in [1..b]$, we associate the tuple $(r_\ell,0,0,\bar{r}_{j_{k_\ell}})$ if $j_{k_\ell} \le i_\ell + 3\tau$, and the tuple $(r_\ell,d_\ell,\tilde{r}_\ell,\bar{r}_{j_{k_\ell}})$ if $j_{k_\ell} > i_\ell{+}3\tau$, where $d_\ell = \pm(t_{\ell}{-}i_\ell - n)$ with plus if $s[t_{\ell}] < s[t_{\ell}{-}p_{\ell}]$ and minus otherwise. We claim that the order of the suffixes $s[i_\ell..n{-}1]$ is the same as the order of their associated tuples and, hence, the suffixes can be sorted by sorting the tuples in $\Oh(n)$ time using the radix sort. We then assemble the SST as above using the LCE index. We do not dive into the proof of the claim since it essentially repeats similar arguments in~\cite{BirenzwigeEtAl}; see~\cite{BirenzwigeEtAl} for details.

\begin{theorem}
For any string of length $n$ over an alphabet $[0..n^{\Oh(1)}]$ and any $b \ge \Omega(\log^2 n)$, one can construct in $\Oh(n\log_b n)$ time and $\Oh(b)$ space on top of the string the sparse suffix tree for arbitrarily chosen $b$ suffixes.\label{thm:sst}
\end{theorem}

\section{\boldmath Refinement of Partitioning Sets}
\label{sec:vishkin-process}

In this section we describe a process that takes the trivial partitioning set $[0..n)$ and iteratively refines it in $\lfloor\log\frac{\tau}{2^4\log^{(3)} n}\rfloor$ phases removing some positions so that, after the $k$th phase, the set is $(2^{k+3}\lfloor\log^{(3)} n\rfloor)$-partitioning and has size $\Oh(n / 2^k)$; moreover, it is ``almost'' $2^{k+3}$-partitioning, satisfying properties~(a) and~(b) but not necessarily~(c) (for $\tau = 2^{k+3}$). In particular, the set after the last phase is $\frac{\tau}{2}$-partitioning 
(and, thus, $\tau$-partitioning by Lemma~\ref{lem:monotone-synch}) 
and has size $\Oh(\frac{n}{\tau} \log^{(3)} n)$. Each phase processes all positions of the currently refined set from left to right and, in an almost online fashion, chooses which of them remain in the set. Rather than performing the phases one after another, which requires $\Oh(n)$ space, we run them simultaneously feeding the positions generated by the $k$th phase to the $(k{+}1)$th phase. Thus, the resulting set is produced in one pass. The set, however, has size $\Oh(\frac{n}{\tau}\log^{(3)} n)$, which is still too large to be stored in $\Oh(n / \tau)$ space; this issue is addressed in Section~\ref{sec:recompression}. Let us elaborate on the details of this process.

Throughout this section, we assume that $\tau \ge 2^5\log^{(3)} n$ 
and, hence, the number of phases is non-zero; the case $\tau < 2^5 \log^{(3)} n$ is addressed in \appendixv{appx:small-tau}{E}. Consider the $k$th phase, for $k \ge 1$. Its input is a set $S_{k-1}$ produced by the $(k{-}1)$th phase; for $k=1$, $S_0 = [0..n)$. Denote by $j_h$ the $h$th position in $S_{k-1}$ (so that $j_1 < \cdots < j_{|S_{k-1}|}$). The phase processes $j_1, j_2, \ldots$ from left to right and decides which of them to put into the new set $S_k \subseteq S_{k-1}$ under construction. The decision for $j_h$ is based on the distances $j_{h} - j_{h-1}$ and $j_{h+1} - j_{h}$, on the substrings $s[j_{h+\ell}..j_{h+\ell}{+}2^k]$ with $\ell \in [-1..4]$, and on certain numbers $v_{h-1}, v_{h}, v_{h+1}$ computed for $j_{h-1}$, $j_{h}$, $j_{h+1}$, which we define below.

For any distinct integers $x, y \ge 0$, denote by $\lbit(x,y)$ the index of the lowest bit in which the bit representations of $x$ and $y$ differ (the lowest bit has index~$0$); e.g., $\lbit(1,0) = 0$, $\lbit(2,8) = 1$, $\lbit(8,0) = 3$. It is well known that $\lbit(x,y)$ can be computed in $\Oh(1)$ time provided $x$ and $y$ occupy $\Oh(1)$ machine words~\cite{Willard}. Denote $\vbit(x,y) = 2\lbit(x,y) + a$, where $a$ is the bit of $x$ with index $\lbit(x,y)$; e.g., $\vbit(8,0) = 7$ and $\vbit(0,8) = 6$. Note that the bit representation of the number $\vbit(x,y)$ is obtained from that of $\lbit(x,y)$ by appending $a$.

Let $w$ be the number of bits in an $\Oh(\log n)$-bit machine word sufficient to represent letters from the alphabet $[0..n^{\Oh(1)}]$ of $s$. For each $j_h$, denote $s_h = \sum_{i=0}^{2^k} s[j_h{+}i] 2^{wi}$. Each number $s_h$ takes $(2^k{+}1)w$ bits and its bit representation coincides with that of the string $s[j_h..j_h{+}2^k]$, when we treat this string as a number stored in memory in the little endian format. The numbers $s_h$ are introduced merely for convenience of the exposition, they are never discerned from their corresponding substrings $s[j_h..j_h{+}2^k]$ in the algorithm. For each $j_h$, define $v'_h = \vbit(s_h, s_{h+1})$ if $j_{h+1} - j_h \le 2^{k-1}$ and $s_h \ne s_{h+1}$, and $v'_h = \infty$ otherwise. Observe that $\lbit(s_h, s_{h+1}) = w\ell + \lbit(s[j_h{+}\ell], s[j_{h+1}{+}\ell])$, where $\ell = \lce(j_h, j_{h+1})$; i.e., $\lbit(s_h, s_{h+1})$ is given by an LCE query in the bit string of length $wn$ obtained from $s$ by substituting each letter with its $w$-bit representation. Define $v''_h = \vbit(v'_h, v'_{h+1}),v'''_h = \vbit(v''_h, v''_{h+1}),v_h = \vbit(v'''_h, v'''_{h+1})$,
assuming $\vbit(x,y) = \infty$ if either $x = \infty$ or $y = \infty$. 

For each $j_h$, denote by $R(j_h)$ a predicate that is true iff $j_{h+1}{-}j_h \le 2^{k-1}$ and $s_h = s_{h+1}$; to verify whether $R(j_h)$ holds, we always check the former condition first and only then the latter if the former condition is satisfied.

\medskip

\noindent \textbf{Refinement rule.}
\emph{The $k$th phase decides to put a position $j_h$ into $S_k$ either if $\infty > v_{h-1} > v_{h}$ and $v_{h} < v_{h+1}$ (i.e., $v_{h-1} \ne \infty$ and $v_{h}$ is a local minimum of the sequence $v_1,v_2,\ldots$), or in three ``boundary'' cases: (i)~$j_{h+1} - j_h > 2^{k-1}$ or $j_h - j_{h-1} > 2^{k-1}$; (ii)~$R(j_{h-1})$ does not hold while $R(j_h)$, $R(j_{h+1})$, $R(j_{h+2})$ hold; (iii)~$R(j_h)$ holds but $R(j_{h+1})$ does not.}

\medskip

For now, assume that the numbers $\lbit(s_h, s_{h+1})$, required to calculate $v'_h$ and $R(j_h)$, are computed by the na{\"i}ve comparison of $s[j_h..j_h{+}2^k]$ and $s[j_{h+1}..j_{h+1}{+}2^k]$ in $\Oh(2^k)$ time (we will change it later). Thus, the process is well defined. The trick with local minima and $\vbit$ reductions is, in essence, as in the deterministic approach of Cole and Vishkin to locally consistent parsings~\cite{ColeVishkin}. In what follows we derive some properties of this approach in order to prove that the $k$th phase indeed produces a $(2^{k+3}\lfloor\log^{(3)} n\rfloor)$-partitioning set.

It is convenient to interpret the $k$th phase as follows (see Fig.~\ref{fig:hills}): the sequence $j_1, j_2, \ldots$ is split into maximal disjoint contiguous regions such that, for any pair of adjacent positions $j_h$ and $j_{h+1}$ inside each region, the distance $j_{h+1} - j_h$ is at most $2^{k-1}$ and $R(j_h) = R(j_{h+1})$. Thus, the regions are of two types: all-$R$ ($\{j_{16}, \ldots, j_{20}\}$ in Fig.~\ref{fig:hills}) and all-non-$R$ ($\{j_{1}, \ldots, j_{15}\}$ or $\{j_{21},\ldots,j_{25}\}$ in Fig.~\ref{fig:hills}). By case~(i), for each long gap $j_{h+1} - j_h > 2^{k-1}$ between regions, we put both $j_h$ and $j_{h+1}$ into $S_k$. In each all-$R$ region, we put into $S_k$ its last position due to case~(iii) and, if the length of the region is at least $3$, its first position by case~(ii). In each all-non-$R$ region, we put into $S_k$ all local minima $v_{h}$ such that $v_{h-1} \ne \infty$. Only all-non-$R$ regions have positions $j_h$ with $v_h \ne \infty$; moreover, as it turns out, only the last three or four their positions $j_h$ have $v_h = \infty$ whereas, for other $j_h$, $v_h \ne \infty$ and $v_h \ne v_{h+1}$. Lemmas~\ref{lem:all-R},~\ref{lem:all-non-R} describe all this formally; their proof is deferred to \appendixv{appx:lcp}{B.1}.

The goal of the fourfold $\vbit$ reduction for $v_h$ is to make $v_h$ small enough so that local minima occur often and, thus, the resulting set $S_k$ is not too sparse. This is the key observation of Cole and Vishkin~\cite{ColeVishkin} and it is stated in Lemma~\ref{lem:v-reduction} and directly follows from the construction of $v_h$ and Lemma~\ref{lem:vishkin}.

\begin{lemma}[{see \cite{ColeVishkin}}]
	Given a string $a_1 a_2 \cdots a_m$ over an alphabet $[0..2^u)$ such that $a_i \ne a_{i+1}$ for any $i \in [1..m)$, the string $b_1 b_2 \cdots b_{m-1}$ such that $b_i = \vbit(a_i, a_{i+1})$, for $i \in [1..m)$, satisfies $b_i \ne b_{i+1}$, for any $i \in [1..m{-}1)$, and $b_i \in [0..2u)$.\label{lem:vishkin}
\end{lemma}
\begin{proof}
Consider $b_i$ and $b_{i+1}$. Denote $\ell = \lbit(a_i, a_{i+1})$ and $\ell' = \lbit(a_{i+1}, a_{i+2})$. As $a_i,a_{i+1} \in [0..2^u)$, we have $\ell \in [0..u)$. Hence, $b_i \le 2\ell + 1 \le 2u - 1$, which proves $b_i \in [0..2u)$. If $b_i = b_{i+1}$, then $\ell = \ell'$ and the bits with indices $\ell$ and $\ell' = \ell$ in $a_i$ and $a_{i+1}$ coincide; on the other hand, by the definition of $\ell = \lbit(a_i, a_{i+1})$, $a_i$ and $a_{i+1}$ must differ in this bit, which is a contradiction.
\end{proof}

\begin{lemma}[{see \cite{ColeVishkin}}]
	For any $v_h \ne \infty$ in the $k$th phase, we have $v_h \in [0..2\log^{(3)} n{+}3)$.\label{lem:v-reduction}
\end{lemma}
\begin{proof}
Since $v'_h \in [0..2nw) = [0..\Oh(n\log n))$, we deduce from Lemma~\ref{lem:vishkin} that $v''_h \in [0..\Oh(\log n))$, $v'''_h \in [0..2\log\log n + \Oh(1))$, and, due to the inequality $\log(x+\delta) \le \log x + \frac{\delta\log e}{x}$, we finally obtain $v_h \in [0..2\log^{(3)} n + 3)$, for sufficiently large $n$.
\end{proof}

The refinement rule implies that, for contiguous regions $j_p, j_{p+1}, \ldots, j_q$ where $R(j_h)$ holds, only $j_p$ and $j_q$ may be in $S_k$ and the period of $s[j_p .. j_q + 2^k]$ is ${\le}2^{k-1}$;  for ``dense'' contiguous regions $j_p, j_{p+1}, \ldots, j_q$ where $R(j_h)$ does not hold, Lemma~\ref{lem:vishkin} ensure frequent local minima. This is summarized in Lemmas~\ref{lem:all-R}, \ref{lem:all-non-R} (the proofs are in \appendixv{appx:lcp}{B.1}).

\begin{restatable}{lemma}{allR}
	Let $j_{p}, j_{p+1}, \ldots, j_{q}$ be a maximal contiguous region of $j_1, j_2, \ldots$ such that, for all $h \in [p..q]$, $R(j_h)$ holds. Then, we have $j_{q} \in S_k$. Further, if $q - p \ge 2$ or $j_p - j_{p-1} > 2^{k-1}$, we have $j_p \in S_k$. All other positions $j_h$ in the region do not belong to $S_k$. The string $s[j_p..j_q{+}2^k]$ has a period at most $2^{k-1}$.\label{lem:all-R}
\end{restatable}

\begin{restatable}{lemma}{allNonR}
	Let $j_{p}, j_{p+1}, \ldots, j_q$ be a maximal contiguous region of $j_1, j_2, \ldots$ such that, for all $h \in [p..q]$, $R(j_h)$ does not hold and, for $h \in [p..q)$, we have $j_{h+1} - j_h \le 2^{k-1}$. Then, $v_h \ne \infty$ for $h \in [p..q{-}4]$, $v_h = \infty$ for $h \in (q{-}3..q]$, and $v_{q-3}$ may be $\infty$ or not. Further, for $h \in [p..q{-}3]$, we have $v_h \ne v_{h+1}$ whenever $v_h \ne \infty$. For $h \in (p..q)$, $j_h \in S_k$ iff $\infty > v_{h-1} > v_{h}$ and $v_{h} < v_{h+1}$; $j_p \in S_k$ iff $j_{p} - j_{p-1} > 2^{k-1}$; $j_q \in S_k$ iff $j_{q+1} - j_q > 2^{k-1}$.\label{lem:all-non-R}
\end{restatable}

\begin{figure}[t]
\center
\includegraphics[width=0.8\textwidth]{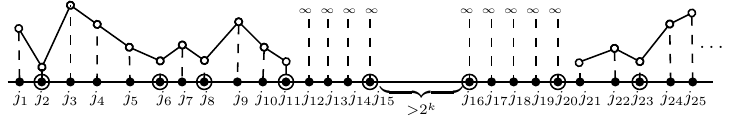}
\caption{The $k$th phase. The heights of the dashed lines over $j_h$ are equal to $v_h$. Encircled positions are put into $S_k$: they are local minima of $v_h$, or are at the ``boundaries'' of all-$R$ regions, or form a gap of length ${>}2^k$. In the figure $R(j_{16}),\ldots, R(j_{20})$ hold and $R(j_{21})$ does not hold.}\label{fig:hills}
\end{figure}

By Lemmas~\ref{lem:all-non-R} and~\ref{lem:v-reduction}, any sequence of $8\log^{(3)} n + 12$ numbers $v_h$ all of which are not $\infty$ contains a local minimum $v_h$ and $j_h$ will be put in $S_k$. Thus, we obtain the following lemma.

\begin{restatable}{lemma}{localDensity}
Let $S_{k-1}$ and $S_k$ be the sets generated by the $(k{-}1)$th and $k$th phases. Then, any range $j_\ell, j_{\ell+1}, \ldots, j_m$ of at least $8\log^{(3)} n + 12$ consecutive positions from $S_{k-1}$ such that $v_h \ne \infty$, for all $h \in [\ell..m]$, has a position from $S_k$.\label{lem:local-density}
\end{restatable}

The following lemma is intuitive but its proof is far from trivial; see \appendixv{appx:refinement-sparsity}{B.2}.
\begin{restatable}{lemma}{localSparsity}
For any $i, i' \in [0..n]$, $|S_k \cap [i..i')| \le 2^6\lceil(i' - i) / 2^{k}\rceil$; in particular, $|S_k| \le n / 2^{k-6}$.\label{lem:local-sparsity}
\end{restatable}

Now we are able to prove that $S_k$ is a $(2^{k+3}\lfloor\log^{(3)} n\rfloor)$-partitioning set and, moreover, it is almost a \mbox{$2^{k+3}$-partitioning} set, in a sense. The proof technique is very similar to the one in~\cite{BirenzwigeEtAl}; for brevity, we defer its detailed proof to \appendixv{appx:refinement-partitioning}{B.2}.

\begin{restatable}{lemma}{partitioningTau}
The $k$th phase generates a $(2^{k+3}\lfloor\log^{(3)} n\rfloor)$-partitioning set $S_k$. Moreover, $S_k$ is almost \mbox{$2^{k+3}$-partitioning:} for $\tau = 2^{k+3}$, it satisfies properties (a) and (b) but not (c), i.e., if $(i..j) \cap S_k = \emptyset$, for $i,j \in S_k$ such that $2^{k+3} < j{-}i \le 2^{k+3}\lfloor\log^{(3)} n\rfloor$, then $s[i..j]$ does not necessarily have period ${\le}2^{k+2}$.\label{lem:2k-partitioning}
\end{restatable}

\section{Speeding up the Refinement Procedure}
\label{sec:time-improvement}

Since, for any $k$, $|S_k| \le n / 2^{k-6}$ by Lemma~\ref{lem:local-sparsity}, it is evident that the algorithm of Section~\ref{sec:vishkin-process} takes $\Oh(|S_0| + |S_1| + \cdots) = \Oh(n)$ time plus the time needed to calculate the numbers $v'_h$, for all positions (from which the numbers $v_h$ are derived). For a given $k \ge 1$, denote by $j_h$ the $h$th position in $S_{k-1}$. For each $j_h$, the number $v'_h$ can be computed by checking whether $j_{h+1} - j_h > 2^{k-1}$ (in this case $v'_h = \infty$), and, if $j_{h+1} - j_h \le 2^{k-1}$, by the na{\"i}ve comparison of $s[j_h..j_h{+}2^k]$ and $s[j_{h+1}..j_{h+1}{+}2^k]$ in $\Oh(2^k)$ time. Thus, all numbers $v'_h$ for the set $S_{k-1}$ can be computed in $\Oh(2^k |S_{k-1}|) = \Oh(n)$ time, which leads to $\Oh(n\log\tau)$ total time for the whole algorithm. This na{\"i}ve approach can be sped up if one can perform the LCE queries that compare $s[j_h..j_h{+}2^k]$ and $s[j_{h+1}..j_{h+1}{+}2^k]$ faster; in fact, if one can do this in $\Oh(1)$ time, the overall time becomes linear. To this end, we exploit the online nature of the procedure. Let us briefly outline the procedure again on a high level.

The algorithm runs simultaneously $\lfloor\log\frac{\tau}{2^4\log^{(3)} n}\rfloor$ phases: the $k$th phase takes positions from the set $S_{k-1}$ produced by the $(k{-}1)$th phase and decides which of them to feed to the $(k{+}1)$th phase, i.e., to put into $S_k$ (the ``top'' phase feeds the positions to an external procedure described in the next section). To make the decision for $j_h \in S_{k-1}$, the $k$th phase needs to know the distance $j_h{-}j_{h-1}$ and the distances $j_{h+\ell}{-}j_h$ to the positions $j_{h+\ell}$ with $\ell \in [1..5]$ such that $j_{h+\ell}{-}j_h \le 5\cdot 2^{k-1}$. Then, the $k$th phase calculates $\min\{2^k{+}1, \lce(j_{h + \ell - 1}, j_{h + \ell})\}$, for all $\ell \in [0..5]$ such that $j_{h + \ell}{-}j_{h + \ell - 1} \le 2^{k-1}$ and $j_{h + \ell}{-}j_h \le 5\cdot 2^{k-1}$, and, based on the distances and the LCE values, computes $v_{h-1}, v_{h}, v_{h+1}$ and decides the fate of $j_h$.

The key for our optimization is the locality of the decision making in the phases that is straightforward for the described process: for any prefix $s[0..d]$, once the positions $S_{k-1} \cap [0..d]$ are known to the $k$th phase, it reports all positions from the set $S_k \cap [0..d{-}5\cdot 2^{k-1}]$ and no position from the set $S_{k-1} \cap [0..d{-}6\cdot 2^{k-1}]$ will be accessed by an LCE query of the $k$th phase in the future. Thus, we can discard all positions $S_{k-1} \cap [0..d{-}6\cdot 2^{k-1}]$ and have to focus only on positions $S_{k-1} \cap (d{-}6\cdot 2^{k-1}..\infty]$ and LCE queries on them in the future.
We deduce from this that after processing the prefix $s[0..d]$ by the whole algorithm, the $k$th phase reports all positions from the set $S_{k} \cap [0..d{-}5\sum_{k'=0}^{k-1} 2^{k'}] \supseteq S_{k} \cap [0..d{-}5\cdot 2^{k}]$ and no LCE query in the $k$th phase accesses positions from the set $S_{k-1} \cap [0..d{-}6\cdot 2^{k}]$ in the future.

This locality of the decision procedure guarantees that, at the time we processed a length-$\ell$ prefix of the string $s$, for some $\ell \ge 0$, all positions from the set $S_k \cap [0..\ell{-}5\cdot 2^k]$ are reported and no position from the set $S_{k-1} \cap [0..\ell{-}5\cdot 2^k]$ will be accessed by an LCE query of the $k$th phase in the future.
Let us summarize this as follows.

\begin{lemma}
Suppose we run the described $\lfloor\log\frac{\tau}{2^4\log^{(3)} n}\rfloor$ phases on a string $s$ of length $n$ from left to right. Then, for any $k \ge 1$ and $d \ge 0$, after processing the prefix $s[0..d]$, the $k$th phase reports all positions from $S_k \cap [0..d{-}5\cdot 2^k]$ to the $(k+1)$th phase and will not perform queries $\lce(j,j')$ on positions $j, j' \in S_{k-1}$ such that $\min\{j,j'\} \le d - 6\cdot 2^k$ in the future.\label{lem:dist-process}
\end{lemma}

Recall that we have $\Oh(\log\tau)$ phases and $\Oh(\log\tau)$ space is always available.
Let us briefly describe main techniques for speeding up the algorithm. Details are given in \appendixv{appx:refinement-speed}{C}.


Suppose that $\tau < \sqrt{n}$ . We have $b = \Theta(\frac{n}{\tau}) \ge \Omega(\sqrt{n})$ additional space for the algorithm in this case. To answer all required LCE queries in constant time, when the algorithm processes a letter $s[d]$, the LCE data structure~\cite{HarelTarjan} is maintained for the leftmost substring $C_i = s[i\lfloor\sqrt{n}\rfloor..(i + 3)\lfloor\sqrt{n}\rfloor - 1]$ whose middle part contains the position $d$ (i.e. $d \in (i+1)\lfloor\sqrt{n}\rfloor..(i + 2)\lfloor\sqrt{n}\rfloor - 1]$). By Lemma~\ref{lem:dist-process}, we can use the data structure to correctly handle all queries because all LCE queries performed by the algorithm at the step $d$ lie within the substring $C_i$. Since we must build the LCE data structure for every $C_i$ once, the overall running time is $\Oh(n + \sum_i |C_i|) = \Oh(n)$ and the occupied space is $\Oh(\sqrt{n}) = \Oh(b)$.

Let us generalize this idea to the case $\tau \ge \sqrt{n}$. Denote $b = \frac{n}{\tau}$. We have $\Oh(b) < \Oh(\sqrt{n})$ space and cannot use the scheme described above since LCE data structures for substrings of length $\Oh(b)$ are not enough to answer queries of the form $\min\{2^k{+}1, \lce(j, j')\}$ when $2^k > \Omega(b)$. The key idea is to group contiguous phases into ``levels'' and maintain SST for a sliding window of positions in each level (in the case $\tau < \sqrt{n}$ we had a single ``level'' and a sliding window of size $\Oh(\sqrt{n})$). We must choose ``level'' size to be large enough to build less SSTs and fit in the $\Oh(n \log_{b} n)$ running time, but also the ``levels'' must be small to efficiently reduce the number of positions in each level and fit all supporting data structures in the $\Oh(b)$ space. To achieve this, we split evenly all $\lfloor\log\frac{\tau}{2^4\log^{(3)} n}\rfloor$ phases into ``levels'', each containing $\Theta(\log\hat{b})$ phases, where $\hat{b} = \lfloor\frac{b}{\log n}\rfloor$. For each ``level'', we maintain a window of $\Oh(\hat{b})$ positions from $S_k$, where $k$ is the lowest phase in the ``level''; one window spans a substring of length $\Oh(2^k \hat{b})$ and the windows change $\Oh(\frac{n}{2^k\hat{b}})$ times in total. Overall we use $\Oh(\hat{b} \log\hat{b}) = \Oh(b)$ space. By Lemma~\ref{lem:2k-partitioning}, the set $S_k$ is ``almost'' $2^k$-partitioning, so we can build SST for each ``level'' as in Lemma~\ref{lem:sst-special} in time $\Oh(2^k\hat{b} + \min\{2^k\hat{b}\log_{\hat{b}} n, \hat{b}\log\hat{b}\})$, which simplifies to $\Oh(\hat{b} \log_{\hat{b}} n)$ for the first ``level'' and to $\Oh(2^k \hat{b})$ for subsequent ``levels''. Overall we can upperbound the running time with $\Oh(n \log_{b} n)$ for all $\Oh(\log_{b} n)$ ``levels''. Thus, the described routine builds partitioning sets $S_k$ in time $\Oh(n\log_{b} n)$ and space $\Oh(b)$. The described sketch of the algorithm is elaborated in details in \appendixv{appx:refinement-speed}{C}.

\section{Recompression}
\label{sec:recompression}

Let $S$ be the set produced by the last phase of the procedure from Sections~\ref{sec:vishkin-process} and~\ref{sec:time-improvement}. By Lemma~\ref{lem:2k-partitioning}, $S$ is a $\frac{\tau}2$-partitioning set of size $\Oh(\frac{n}{\tau} \log^{(3)} n)$. Throughout this section, we assume that $\tau \ge (\log^{(3)} n)^4$ so that the size of $S$ is at most $\Oh(\frac{n}{(\log^{(3)} n)^3})$; the case $\tau < (\log^{(3)} n)^4$ is discussed in \appendixv{appx:small-tau}{E}. In what follows we describe an algorithm that removes positions from $S$ transforming it into a $\tau$-partitioning set of size $\Oh(n / \tau)$.

Instead of storing $S$ explicitly, which is impossible in $\Oh(n / \tau)$ space, we construct a related-to-$S$ string $R$ of length $\Oh(\frac{n}{\tau} \log^{(3)} n)$ over a small alphabet such that $R$ can be packed into $\Oh(n / \tau)$ machine words. Positions of $S$ are represented, in a way, by letters of $R$. The construction of $R$ is quite intricate, which is necessary in order to guarantee that letters of $R$ corresponding to close positions of $S$ (namely, positions at a distance at most $\tau / 2^5$) are necessarily distinct even if the letters are not adjacent in $R$. This requirement is stronger than the requirement of distinct adjacent letters that was seen, for instance, in Lemma~\ref{lem:vishkin} but it is achieved by similar means using $\vbit$ reductions as in Section~\ref{sec:vishkin-process}. We then apply to $R$ a variant of the iterative process called \emph{recompression}~\cite{Jez4} that removes some letters thus shrinking the length of $R$ to $\Oh(n / \tau)$. Then, the whole procedure of Sections~\ref{sec:vishkin-process}--\ref{sec:time-improvement} that generated $S$ is performed again but this time we discard all positions of $S$ corresponding to removed positions of the string $R$ and store the remaining positions explicitly in a set $S^* \subseteq S$. We show that $S^*$ is $\tau$-partitioning and has size $\Oh(n / \tau)$. Let us elaborate on the details.


The algorithm starts with an empty string $R$ and receives positions of $S$ from left to right appending to the end of $R$ new letters corresponding to the received positions. It is more convenient to describe the algorithm as if it acted in two stages: the first stage produces a $\frac{3}{4}\tau$-partitioning set $S' \subseteq S$, for which a condition converse to property~(c) holds (thus, some positions of $S$ are discarded already in this stage), and the second stage, for each position of $S'$, appends to the end of $R$ a letter of size $\Oh((\log^{(3)} n)^2)$ bits. Both stages act in an almost online fashion and, hence, can be actually executed simultaneously in one pass without the need to store the auxiliary set $S'$. The separation is just for the ease of the exposition.

\subparagraph{\boldmath The first stage.}
The goal is to construct set $S' \subseteq S$ by excluding from $S$ all positions $h$ for which there exist $i,j \in S$ such that $i < h \le j$, $j - i \le \tau / 4$, and $s[i..i{+}\tau/2] = s[j..j{+}\tau/2]$. The algorithm generating $S'$ is as follows.

We consider all positions of $S$  from left to right and, for each $i \in S$, process every $j \in (i..i{+}\tau/4] \cap S$ by comparing $s[i..i{+}\tau/2]$ with $s[j..j{+}\tau/2]$. If $s[i..i{+}\tau/2] = s[j..j{+}\tau/2]$, then we traverse all positions of the set $(i..j] \cap S$ from right to left marking them for removal until an already marked position is encountered. Since the marking procedure works from right to left, every position is marked at most once. The position $i$ is put into $S'$ iff it was not marked previously. During the whole process, we maintain a ``look-ahead'' queue that stores the positions $(i..i{+}\tau/4] \cap S$ and indicates which of them were marked for removal. 

Due to Lemma~\ref{lem:local-sparsity}, the size of the set $(i..i{+}\tau/4] \cap S$ is $\Oh(\log^{(3)} n)$. Therefore, the look-ahead queue takes $\Oh(\log^{(3)} n)$ space, which is $\Oh(n / \tau)$ since $n / \tau \ge \log^2 n$, and $\Oh(\log^{(3)} n)$ comparisons are performed for each $i$. Hence, if every comparison takes $\Oh(1)$ time, the set $S'$ is constructed in $\Oh(|S|\log^{(3)} n) = \Oh(\frac{n}{\tau}(\log^{(3)} n)^2)$ time, which is $\Oh(n)$ since $\tau \ge (\log^{(3)} n)^4$. Thus, it remains to explain how the comparisons can be performed.

Similar to the algorithm of Section~\ref{sec:time-improvement}, we consecutively consider substrings $C'_i = s[i\tau..(i + 3)\tau)$, for $i \in [0..n / \tau - 3]$: when all positions from a set $S \cap [i\tau..(i + 3)\tau)$ are collected, we use the algorithm of Lemma~\ref{lem:sst-special} to build a SST for all suffixes of the string $C'_i$ whose starting positions are from $S$; the tree, endowed with an LCA data structure~\cite{HarelTarjan}, is used in the procedure for deciding which of the positions from the set $S \cap [(i + \frac{1}2)\tau .. (i + \frac{3}2)\tau)$ (or $S \cap [0 .. \frac{3}{2}\tau)$ if $i = 0$) should be marked for removal. Thus, after processing the last string $C'_i$, all positions of $S$ are processed and $S'$ is generated. By Lemma~\ref{lem:local-sparsity}, the number of suffixes in the SST for $C'_i$ is $\Oh(\log^{(3)} n)$ and, therefore, the tree occupies $\Oh(\log^{(3)} n) \le \Oh(n / \tau)$ space and its construction takes $\Oh(\tau + \log^{(3)} n \cdot \log \log^{(3)} n)$ time by Lemma~\ref{lem:sst-special}, which is $\Oh(\tau)$ since $\tau \ge (\log^{(3)} n)^4$. Thus, the total construction time for all the trees in the stage is $\Oh(\frac{n}{\tau}\tau) = \Oh(n)$ and the space used is $\Oh(\log^{(3)} n)$ since, at every moment, at most one tree is maintained.

The following lemma shows that the transformation within the first stage does not break $\tau$-partitioning properties. Its proof is deferred to \appendixv{appx:first-stage}{D.1}.
\begin{restatable}{lemma}{periodicityGap}
The set $S'$ is $\tau$-partitioning and satisfies a converse of property~(c): if a substring $s[i..j]$ has a period at most $\tau / 4$, then $S' \cap [i{+}\frac{3}{4}\tau .. j{-}\frac{3}{4}\tau] = \emptyset$. Moreover, $S'$ is almost $\frac{3}{4}\tau$-partitioning, meeting properties~(a) and~(b) with $\frac{3}{4}\tau$ in place of $\tau$, but not necessarily~(c).\label{lem:periodicity-gap}
\end{restatable}

\subparagraph{The second stage.}
We consider all positions of $S'$ from left to right and, for each $p \in S'$, append to the end of the (initially empty) string $R$ a new carefully constructed letter $a_p$ occupying $\Oh((\log^{(3)} n)^2)$ bits. Thus, the string $R$ will have length $|S'|$ and will take $\Oh(|S'|(\log^{(3)} n)^2) = \Oh(\frac{n}{\tau} (\log^{(3)} n)^3)$ bits of space, which can be stored into $\Oh(n / \tau)$ machine words of size $\Oh(\log n)$ bits. The crucial property of $R$ for us is that any two letters of $R$ corresponding to close positions of $S'$ are distinct, namely the following lemma will be proved:

\begin{lemma}
For any $p, \bar{p} \in S'$, if $0 < \bar{p} - p \le \tau / 2^5$, then $a_p \ne a_{\bar{p}}$.\label{lem:r-letters-equality}
\end{lemma}

Consider $p \in S'$. We are to describe an algorithm generating an $\Oh((\log^{(3)} n)^2)$-bit letter $a_p$ for $p$ that will be appended to the string $R$.

Denote by $p_1, p_2, \ldots, p_{m}$ all positions of $S' \cap (p..p{+}\tau/2^5]$ in the increasing order. By Lemma~\ref{lem:local-sparsity}, $m \le \Oh(\log^{(3)} n)$ and, hence, there is enough space to store them. By construction, $s[p..p{+}\frac{\tau}{2}] \ne s[p_j..p_j{+}\frac{\tau}{2}]$, for each $j \in [1..m]$. One can compute the longest common prefix of $s[p..p{+}\frac{\tau}{2}]$ and $s[p_j..p_j{+}\frac{\tau}{2}]$, for any $j \in [1..m]$, in $\Oh(1)$ time using a SST with an LCA data structure~\cite{HarelTarjan} built in the first stage for a substring $C'_i = s[i\tau..(i + 3)\tau - 1]$ such that $p \in [i\tau..(i+\frac{3}{2})\tau)$. 
(In order to have $p_1, p_2, \ldots, p_m$ prepared, we handle $p$, which was reported by the first stage after processing $C'_i$, only when $C'_{i+1}$ was processed too; thus, the first stage maintains two SSTs: one for a substring $C'_{i+1}$ currently under analysis and one for $C'_{i}$, retained for its use in the second stage.)

Denote $\ell = 2^6\lceil\log^{(3)} n\rceil$. Recall that $S$ is produced by the $k$th phase of the procedure of Section~\ref{sec:vishkin-process}, for $k = \lfloor\log\frac{\tau}{2^4\log^{(3)} n}\rfloor$, and hence, by Lemma~\ref{lem:local-sparsity}, the size of any set $S \cap [i..j]$, for $i \le j$, is at most $2^6\lceil (j - i + 1) / 2^{k}\rceil$. Therefore, since $S' \subseteq S$ and $m$ is the size of the set $S' \cap (p..p{+}\tau/2^5]$, we obtain $m \le 2^6 (\tau / 2^5) / \frac{\tau}{2\cdot 2^4\log^{(3)} n} \le \ell$.

Let $w$ be the number of bits in an $\Oh(\log n)$-bit machine word sufficient to represent letters from the alphabet $[0..n^{\Oh(1)}]$ of $s$. For each $p_j$, denote $t_j = \sum_{i=0}^{\tau/2} s[p_j{+}i] 2^{wi}$; similarly, for $p$, denote $t = \sum_{i=0}^{\tau/2} s[p{+}i] 2^{wi}$. As in an analogous discussion in Section~\ref{sec:vishkin-process}, we do not discern the numbers $t_j$ and $t$ from their corresponding substrings in $s$ and use them merely in the analysis. The intuition behind our construction is that the numbers $t, t_1, t_2, \ldots, t_m$, in principle, could have been used for the string $R$ as letters corresponding to the positions $p, p_1, p_2, \ldots, p_m$ since $t, t_1, t_2, \ldots, t_m$ are pairwise distinct (due to the definition of $S'$) but, unfortunately, they occupy too much space ($\Oh(w\tau)$ bits each). One has to reduce the space for the letters retaining the property of distinctness. The tool capable to achieve this was already developed in Section~\ref{sec:vishkin-process}: it is the $\vbit$ reduction, a trick from Cole and Vishkin's deterministic locally consistent parsing~\cite{ColeVishkin}.

We first generate for $p$ a tuple of $\ell$ numbers $\langle w'_1, w'_2, \ldots, w'_\ell\rangle$: for $j \in [1..\ell]$, $w'_j = \vbit(t, t_j)$ if $j \le m$, and $w'_j = \infty$ otherwise. Since the longest common prefix of substrings $s[p..p{+}\frac{\tau}{2}]$ and $s[p_j..p_j{+}\frac{\tau}{2}]$, for $j \in [1..m]$, can be calculated in $\Oh(1)$ time, the computation of the tuple takes $\Oh(\ell) = \Oh(\log^{(3)} n)$ time. By Lemma~\ref{lem:vishkin}, each number $w'_j$ occupies less than $\lceil\log w + \log\tau + 1\rceil$ bits. Thus, we can pack the whole tuple into $\ell\lceil\log w + \log\tau + 1\rceil$ bits encoding each value $w'_j$ into $\lceil\log w + \log\tau + 1\rceil$ bits and representing $\infty$ by setting all bits to~$1$. We denote this chunk of $\ell\lceil\log w + \log\tau + 1\rceil$ bits by $\bar{t}$. In the same way, for each $p_i$ with $i \in [1..m]$, we generate a tuple $\langle w'_{i,1}, w'_{i,2}, \ldots, w'_{i,\ell}\rangle$ comparing $s[p_i..p_i{+}\tau/2]$ to $s[q..q{+}\tau/2]$, for each $q \in S' \cap (p_i..p_i{+}\tau/2^5]$, and using the $\vbit$ reduction; the tuple is packed into a chunk $\bar{t}_i$ of $\ell\lceil\log w + \log\tau + 1\rceil$ bits. See Figure~\ref{fig:w-scheme}. For each $j \in [1..m]$, the number $w'_j$ is not equal to $\infty$ and, thus, due to Lemma~\ref{lem:vishkin}, differs from the number $w'_{j,j}$ (the $j$th element of the tuple $\langle w'_{i,1}, w'_{i,2}, \ldots, w'_{i,\ell}\rangle$). Therefore, all the tuples---and, hence, their corresponding numbers $\bar{t}, \bar{t}_1, \bar{t}_2, \ldots, \bar{t}_m$---are pairwise distinct.

\newcommand{\rn}[2]{
    \tikz[remember picture,baseline=(#1.base)]\node [inner sep=0] (#1) {$#2$};%
}

\begin{figure}[t]
\begin{equation*}
\begingroup
\renewcommand*{\arraystretch}{1.5}
\begin{matrix}
\rn{m11}{t}               & \rn{m12}{t_1}                 & \rn{m13}{t_2}                        & \ldots  & \rn{m14}{t_m}              & \rn{m15}{t_{m+1}}              & \rn{m16}{t_{m+2}} & \rn{m17}{t_{m+3}}  &
\rn{m18}{t_{m+4}} & \rn{m19}{t_{m+5}}
\\[5pt]
\rn{m21}{\bar{t}}        & \rn{m22}{\bar{t}_1}             & \rn{m23}{\bar{t}_2}                 & \ldots  & \rn{m24}{\bar{t}_m}         & \rn{m25}{\bar{t}_{m+1}}       & \rn{m26}{\bar{t}_{m+2}}       & \rn{m27}{\bar{t}_{m+3}}
\\[5pt]
\rn{m31}{\bar{\bar{t}}}  & \rn{m32}{\bar{\bar{t}}_1}      & \rn{m33}{\bar{\bar{t}}_2}            & \ldots  & \rn{m34}{\bar{\bar{t}}_m}    & \rn{m35}{\bar{\bar{t}}_{m+1}}
\\[5pt]
\rn{m41}{\bar{\bar{\bar{t}}}}  & \rn{m42}{\bar{\bar{\bar{t}}}_1} & \rn{m43}{\bar{\bar{\bar{t}}}_2} & \ldots & \rn{m44}{\bar{\bar{\bar{t}}}_m}
\\[5pt]
\rn{m51}{a_p}
\end{matrix}
\endgroup
\begin{tikzpicture}[overlay,remember picture]
\draw [->] (m11) -- (m21);\draw [->] (m21) -- (m31);\draw [->] (m31) -- (m41);\draw [->] (m41) -- (m51);
\draw [->] (m12) -- (m22);\draw [->] (m22) -- (m32);\draw [->] (m32) -- (m42);
\draw [->] (m13) -- (m23);\draw [->] (m23) -- (m33);\draw [->] (m33) -- (m43);
\draw [->] (m14) -- (m24);\draw [->] (m24) -- (m34);\draw [->] (m34) -- (m44);
\draw [->] (m15) -- (m25);\draw [->] (m25) -- (m35);
\draw [->] (m16) -- (m26);
\draw [->] (m17) -- (m27);
\draw [->] (m12) -- (m21);\draw [->] (m13) -- (m21);\draw [->] (m14) -- (m21);
\draw [->] (m13) -- (m22);\draw [->] (m14) -- (m22);
\draw [->] (m14) -- (m23);\draw [->] (m15) -- (m23);
\draw [->] (m15) -- (m24);\draw [->] (m15) -- (m24);
\draw [->] (m16) -- (m25);\draw [->] (m17) -- (m25);
\draw [->] (m17) -- (m26);
\draw [->] (m18) -- (m27);\draw [->] (m19) -- (m27);
\draw [->] (m22) -- (m31);\draw [->] (m23) -- (m31);\draw [->] (m24) -- (m31);
\draw [->] (m23) -- (m32);\draw [->] (m24) -- (m32);
\draw [->] (m24) -- (m33);\draw [->] (m25) -- (m33);
\draw [->] (m25) -- (m34);\draw [->] (m25) -- (m34);
\draw [->] (m26) -- (m35);\draw [->] (m27) -- (m35);
\draw [->] (m32) -- (m41);\draw [->] (m33) -- (m41);\draw [->] (m34) -- (m41);
\draw [->] (m33) -- (m42);\draw [->] (m34) -- (m42);
\draw [->] (m34) -- (m43);\draw [->] (m35) -- (m43);
\draw [->] (m35) -- (m44);\draw [->] (m35) -- (m44);
\draw [->] (m42) -- (m51);
\draw [->] (m43) -- (m51);
\draw [->] (m44) -- (m51);
\end{tikzpicture}
\end{equation*}
\caption{The scheme generating $a_p$ via $\vbit$ reductions. If a node $\hat{t}$ has ingoing edges labeled with $\tilde{t}, \tilde{t}_1, \tilde{t}_2, \ldots, \tilde{t}_r$ (from left to right), then $\hat{t}$ encodes a tuple $\langle \tilde{w}_1, \tilde{w}_2, \ldots, \tilde{w}_\ell \rangle$ such that, for $j \in [1..r]$, $\tilde{w}_j = \vbit(\tilde{t}, \tilde{t}_j)$ and, for $j \in (r..\ell]$, $\tilde{w}_j = \infty$. In the figure, the numbers $t, t_1, t_2, \ldots, t_{m+5}$ correspond to consecutive positions $p, p_1, p_2, \ldots, p_{m+5}$ in the set $S'$, respectively. By looking at which of the ingoing edges are present and which are not, one can deduce that here we have $S' \cap (p..p{+}\tau/2^5] = \{p_1, \ldots, p_m\}$, $S' \cap (p_1..p_1{+}\tau/2^5] = \{p_2, \ldots, p_m\}$, $S' \cap (p_2..p_2{+}\tau/2^5] = \{p_3, \ldots, p_m, p_{m+1}\}$, $S' \cap (p_m..p_m{+}\tau/2^5] = \{p_{m+1}\}$, $S' \cap (p_{m+1}..p_{m+1}{+}\tau/2^5] = \{p_{m+2}, p_{m+3}\}$, $S' \cap (p_{m+2}..p_{m+2}{+}\tau/2^5] = \{p_{m+3}\}$, $S' \cap (p_{m+3}..p_{m+3}{+}\tau/2^5] = \{p_{m+4}, p_{m+5}\}$.}
\label{fig:w-scheme}
\end{figure}
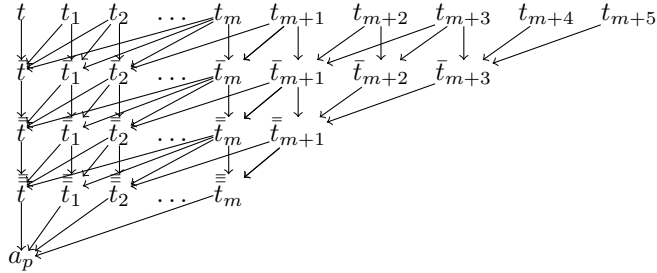

The numbers $\bar{t}, \bar{t}_1, \bar{t}_2, \ldots, \bar{t}_m$, like the numbers $t, t_1, t_2, \ldots, t_m$, could have been used, in principle, as letters for the string $R$ but they still are too large. We therefore repeat the same $\vbit$ reduction but now for the numbers $\bar{t}, \bar{t}_1, \bar{t}_2, \ldots, \bar{t}_m$ in place of $t, t_1, t_2, \ldots, t_m$ thus generating a tuple $\langle w''_1, w''_2, \ldots, w''_\ell\rangle$: for $j \in [1..\ell]$, $w''_j = \vbit(\bar{t}, \bar{t}_j)$ if $j \le m$, and $w''_j = \infty$ otherwise. The computation of $\vbit(\bar{t}, \bar{t}_j)$ takes $\Oh(\ell)$ time since $\bar{t}$ occupies $\ell$ machine words of size $\Oh(\log n)$ bits. It follows from Lemma~\ref{lem:vishkin} that the tuple $\langle w''_1, w''_2, \ldots, w''_\ell\rangle$ can be packed into a chunk $\bar{\bar{t}}$ of $\ell\lceil\log\ell + \log\lceil\log w + \log\tau + 1\rceil + 1\rceil$ bits (i.e., $\Oh(\log^{(3)} n \cdot \log \log n)$ bits), which already fits into one machine word. We perform analogous reductions for the positions $p_1, p_2, \ldots, p_m$ generating $m$ tuples $\langle w''_{i,1}, w''_{i,2}, \ldots, w''_{i,\ell}\rangle$, for $i\in [1..m]$, packed into new chunks $\bar{\bar{t}}_1, \bar{\bar{t}}_2, \ldots, \bar{\bar{t}}_m$, respectively. Note that, in order to produce a tuple $\langle w''_{i,1}, w''_{i,2}, \ldots, w''_{i,\ell}\rangle$, for $i \in [1..m]$, that is packed into $\bar{\bar{t}}_i$, we use not only the numbers $\bar{t}_i, \bar{t}_{i+1}, \ldots, \bar{t}_m$ corresponding to positions $p_i, p_{i+1}, \ldots, p_m$ but also similarly computed numbers at other positions from $S' \cap (p_i..p_i{+}\tau/2^5]$, if any. See Figure~\ref{fig:w-scheme} for a clarification: it can be seen that the ``top'' numbers include not only $t, t_1, \ldots, t_m$ precisely because of this.

By the same argument that proved the distinctness of $\bar{t}, \bar{t}_1, \bar{t}_2, \ldots, \bar{t}_m$, one can easily show that $\bar{\bar{t}}, \bar{\bar{t}}_1, \bar{\bar{t}}_2, \ldots, \bar{\bar{t}}_m$ are pairwise distinct. But they are still too large to be used as letters of $R$. Then again, we repeat the same reductions at positions $p, p_1, p_2, \ldots, p_m$ but now for the numbers $\bar{\bar{t}}, \bar{\bar{t}}_1, \bar{\bar{t}}_2, \ldots, \bar{\bar{t}}_m$ in place of $\bar{t}, \bar{t}_1, \bar{t}_2, \ldots, \bar{t}_m$, thus generating new chunks $\bar{\bar{\bar{t}}}, \bar{\bar{\bar{t}}}_1, \bar{\bar{\bar{t}}}_2, \ldots, \bar{\bar{\bar{t}}}_m$. Finally, once more, we do the $\vbit$ reduction for $\bar{\bar{\bar{t}}}, \bar{\bar{\bar{t}}}_1, \bar{\bar{\bar{t}}}_2, \ldots, \bar{\bar{\bar{t}}}_m$ generating a tuple $\langle w_1, w_2, \ldots, w_\ell\rangle$ such that, for $j \in [1..\ell]$, $w_j $ is $\vbit(\bar{\bar{\bar{t}}}, \bar{\bar{\bar{t}}}_j)$ if $j \le m$, and $\infty$ otherwise.

Using the same reasoning as in the proof of Lemma~\ref{lem:v-reduction}, one can deduce from Lemma~\ref{lem:vishkin} that the tuple $\langle w_1, w_2, \ldots, w_\ell\rangle$ fits into a chunk of $\ell \cdot 2 \log \log^{(3)} n \le 2^6 \lceil\log^{(3)} n\rceil^2$ bits (the inequality holds provided $n > 2^{16}$) encoding each value $w_j$ into $\lceil\log^{(3)} n\rceil$ bits and representing $\infty$ by setting all $\lceil\log^{(3)} n\rceil$ bits to~$1$. Denote by $a_p$ this chunk of $2^6\lceil\log^{(3)} n\rceil^2$ bits that encodes the tuple. We treat $a_p$ as a new letter of $R$ that corresponds to the position $p$ and we append $a_p$ to the end of $R$. Lemma~\ref{lem:r-letters-equality} follows then straightforwardly by construction.

Given $p \in S'$, the calculation of the numbers $\bar{t}, \bar{\bar{\bar{t}}}, a_p$ takes $\Oh(\ell^2)$ time. The calculation of $\bar{\bar{t}}$ requires $\Oh(\ell^3)$ time since each reduction $\vbit(\bar{t}, \bar{t}_j)$ for it takes $\Oh(\ell)$ time. Hence, the total time for the construction of $R$ is $\Oh(|S'|\ell^3) = \Oh(\frac{n}{\tau} (\log^{(3)} n)^4)$, which is $\Oh(n)$ as $\tau \ge (\log^{(3)} n)^4$.


\newcommand{\LB}{6}
\newcommand{\LBB}{5}
\newcommand{\BB}{11}
\newcommand{\BBB}{10}
\newcommand{\LBplusBBB}{16}

\subparagraph{Recompression.}
If the distance between any pair of adjacent positions of $S'$ is at least $\tau / 2^{\LB}$, then $|S'| \le 2^{\LB} n / \tau$ and, by Lemma~\ref{lem:periodicity-gap}, $S'$ can be used as the resulting $\tau$-partitioning set of size $\Oh(n / \tau)$. Unfortunately, in general, this is not the case and we have to ``sparsify'' $S'$.

There is a one-to-one correspondence between $S'$ and positions of $R$. Using a technique of Je{\.z}~\cite{Jez4} called \emph{recompression} 
, we can remove in $\Oh(|R|)$ time some letters of $R$ reducing by a fraction $\frac{4}{3}$ the number of pairs of adjacent letters $R[i], R[i{+}1]$ whose corresponding positions in $S'$ are at a distance at most $\tau / 2^{\LB}$.
We perform such reductions until the length of $R$ becomes at most $2^{14} \cdot n / \tau$. The positions of $S'$ corresponding to remaining letters will constitute a $\tau$-partitioning set of size $\Oh(n / \tau)$. In order to guarantee that this subset of $S'$ is $\tau$-partitioning, we have to execute the recompression reductions gradually increasing the distances that are of interest for us: first, we get rid of adjacent pairs with distances at most $\tau / \log^{(3)} n$ between them, then the threshold is increased to $2\tau / \log^{(3)} n$, then $2^2\tau / \log^{(3)} n$, and so on until (most) adjacent pairs with distances at most $2^{\log^{(4)} n - \LB} \tau / \log^{(3)} n = \tau / 2^{\LB}$ between them are removed in last recompression reductions. The details follow.


Since it is impossible to store in $\Oh(n / \tau)$ space the precise distances between positions of $S'$, the information about distances needed for recompression is encoded as follows. For each $i \in [0..|R|)$ and a position $p \in S'$ corresponding to the letter $R[i]$, we store an array of numbers $M_i[0..\lceil\log^{(4)} n\rceil]$ such that, for $j \in [0..\lceil\log^{(4)} n\rceil]$, $M_i[j]$ is equal to the size of the set $S' \cap (p..p{+}\tau / 2^j]$. By Lemma~\ref{lem:local-sparsity}, we have $|S' \cap (p..p{+}\tau]| \le \Oh(\log^{(3)} n)$ and, hence, each number $M_i[j]$ occupies $\Oh(\log^{(4)} n)$ bits. Therefore, all the arrays $M_i$ can be stored in $\Oh(|R|(\log^{(4)} n)^2) \le \Oh(\frac{n}{\tau}\log^{(3)} n \cdot (\log^{(4)} n)^2)$ bits, which fits into $\Oh(\frac{n}{\tau})$ machine words of size $\Oh(\log n)$ bits. All arrays $M_i$ are constructed in a straightforward way in $\Oh(|R|\log^{(3)} n) = \Oh(\frac{n}{\tau}(\log^{(3)} n)^2)$ time (which is $\Oh(n)$ since $\tau \ge (\log^{(3)} n)^4$) during the left-to-right pass over $S'$ that generated the string $R$.

Our algorithm consecutively considers all numbers $j \in [\LB..\lceil\log^{(4)} n\rceil]$ in decreasing order, starting from $j = \lceil\log^{(4)} n\rceil$. For each $j$, it iteratively performs a recompression procedure reducing the number of adjacent letters $R[i], R[i{+}1]$ whose corresponding positions from $S'$ are at a distance at most $\tau / 2^j$, until $R$ shrinks to a length at most $2^{j+\BBB}{\cdot}\frac{n}{\tau}$.
Thus, $|R| \le 2^{\LBplusBBB}{\cdot}\frac{n}{\tau}$ after last recompression reductions for $j = \LB$. Let us describe the recompression procedure.

Fix $j \in [\LB..\lceil\log^{(4)} n\rceil]$. To preserve property~(c) of the $\tau$-partitioning set $S'$ during the sparsifications, we impose an additional restriction: a letter $R[i]$ cannot be removed if either $i = 0$ or the distance between the position $p \in S'$ corresponding to $R[i]$ and the predecessor of $p$ in $S'$ is larger than $\tau / 2^5$, i.e., if $M_{i-1}[5] = 0$. The rationale is as follows: the position $p$ might be the right boundary of a gap in $S'$ of length $>\tau$ and it is dangerous to break the gap since, once $p$ is removed, the gap might not satisfy property~(c) (the range of the string $s$ corresponding to the gap should have a period that is at most $\tau/4$).


The processing of the number $j$ starts with checking whether $|R| \le 2^{j+\BBB} \cdot \frac{n}{\tau}$. If so, we skip the processing of $j$ and move to $j - 1$ (provided $j > \LB$). Suppose that $|R| > 2^{j+\BBB} \cdot \frac{n}{\tau}$. Denote $\sigma = 2^{2^6\lceil\log^{(3)} n\rceil^2}$, the size of the alphabet $[0..\sigma)$ of $R$.  Then, the algorithm creates an array $P[0..\sigma{-}1][0..\sigma{-}1]$ filled with zeros, which occupies $\Oh(\sigma^2) = \Oh(2^{2^7(\log^{(3)} n)^2}) = o(\log n)$ space, and collects in $P$ statistics on pairs of adjacent letters of $R$ whose corresponding positions in $S'$ are at a distance at most $\tau / 2^j$ and whose first letter may be removed: namely, we traverse all $i \in [1..|R|)$ and, if $M_i[j] \ne 0$ and $M_{i-1}[5] \ne 0$, then we increase by one the number $P[R[i]][R[i{+}1]]$. By Lemma~\ref{lem:r-letters-equality}, $R[i] \ne R[i+1]$ when $M_i[j] \ne 0$.

The core tool of the recompression technique proposed by Je{\.z}~\cite{Jez4} is an algorithm for multidigraph without self-loops $G = (V, E)$ that constructs a directed cut of size at least $\lceil \frac{|E|}{4} \rceil$ edges in time $\Oh(|V|^2)$ if the graph is given by an adjacency matrix. If we interpret $P$ as an adjacency matrix, we can use Je{\.z}'s technique (there are no self-loops because $R[i] \ne R[i + 1]$ when $M_i[j] \ne 0$ due to Lemma~\ref{lem:r-letters-equality}) and split the alphabet into two disjoint subsets correspoding to the cut: $[0..\sigma) = \acute{\Sigma} \sqcup \grave{\Sigma}$. After that we mark for removal from $R$ all indices $i \in [1..|R|{-}1)$ for which the following conditions hold: $M_i[j] \ne 0$, $M_{i-1}[5] \ne 0$, $R[i] \in \acute{\Sigma}$, and $R[i{+}1] \in \grave{\Sigma}$. Once the sets $\acute{\Sigma}$ and $\grave{\Sigma}$ are computed in time $\Oh(\sigma^2) = o(\log^2 n)$, the marking takes $\Oh(|R|)$ time and can be organized using a bit array of length $|R|$.

After the marking step we update values in all arrays $M_i$ according to removal marks in one right to left pass: for each $i \in [0..|R|)$ and $j'\in [0..\lceil\log^{(4)} n\rceil]$, the new value for $M_i[j']$ is the number of indices $i + 1, i + 2, \ldots, i + M_i[j']$ that were not marked for removal, i.e., $M_i[j']$ is the number of positions in the set $S' \cap (p..p{+}\tau/2^{j'}]$ whose corresponding letters $R[i']$ will remain in $R$, where $p \in S'$ is the position corresponding to $R[i]$. Since $M_i[j'] \le M_{i+1}[j'] + 1$, for $i \in [0..|R|{-}1)$, the pass updating $M$ can be executed in $\Oh(|R|\log^{(4)} n)$ time.

Finally, we delete letters $R[i]$ and arrays $M_i$, for all indices $i$ marked for removal, thus shrinking the length of $R$ and the storage for $M_i$. We call this procedure, which marks letters of $R$ and removes them and their corresponding arrays $M_i$, the recompression. One recompression iteration takes $\Oh(|R|\log^{(4)} n)$ time, where $|R|$ is the length of $R$ before shrinking.

The next lemma states that the recompression shrinks the string $R$ by a constant factor. 
\begin{lemma}
If, for $j \in [\LB..\lceil\log^{(4)} n\rceil]$, before the recompression procedure there were $d$ non-zero numbers $M_i[j]$ with $i \in [1..|R|)$ such that $M_{i-1}[5] \ne 0$, then the arrays $M_i$ modified by the procedure, for all $i$ corresponding to unremoved positions of $R$, contain at most $\frac{3}{4}d$ non-zero numbers $M_i[j]$ such that $M_{i-1}[5]{\ne}0$.\label{lem:recompression}%
\begin{proof}
The proof repeats an argument from~\cite{Jez4} and~\cite[Lemma 7]{I}. Consider an undirected weighted graph $G$ corresponding to the digraph encoded in the adjacency matrix $P$. By construction of $P$, we have $d = \sum_{a\ne b} P[a][b]$, which follows from Lemma~\ref{lem:r-letters-equality} that guarantees $R[i] \ne R[i+1]$ when $M_i[j] \ne 0$.
Thus, $d$ is the sum of weights of all edges in $G$.
Putting a letter $a$ into either $\acute{\Sigma}$ or $\grave{\Sigma}$, we add to the cut at least half of the total weight of all edges connecting $a$ to the letters $0,1,\ldots,a{-}1$.
Therefore, the cut of $G$ induced by $\acute{\Sigma}$ and $\grave{\Sigma}$ has a weight at least $\frac{1}2 d$. The edges in the cut might be directed both from $\acute{\Sigma}$ to $\grave{\Sigma}$ and in the other direction.
Switching $\acute{\Sigma}$ and $\grave{\Sigma}$, if needed, we ensure that the direction from $\acute{\Sigma}$ to $\grave{\Sigma}$ has a maximal total weight, which is obviously at least $\frac{1}4 d$. According to this cut, we mark for removal from $R$ at least $\frac{1}{4} d$ letters $R[i]$ such that $M_i[j] \ne 0$. Hence, the number of non-zero values $M_i[j]$ such that $M_{i-1}[5] \ne 0$ is reduced by $\frac{1}{4} d$, which gives the result of the lemma since new non-zero values could not appear after the deletions.
\end{proof}
\end{lemma}

Suppose, for a fixed $j \in [\LB..\lceil\log^{(4)} n\rceil]$, the algorithm has performed one iteration of the recompression. Denote by $S''$ the set of all positions from $S'$ that ``survived'' the recompression for $j \in [\LB..\lceil\log^{(4)} n\rceil]$ and, thus, have a corresponding letter in the updated string $R$. There is a one-to-one correspondence between $S''$ and letters of $R$. For each $i \in [0..|R|)$ and $j' \in [0..\lceil\log^{(4)} n\rceil]$, the number $M_i[j']$ in the modified arrays $M_i$ is the size of the set $S'' \cap (p..p{+}\tau/2^{j'}]$, for a position $p \in S''$ corresponding to $i$. We therefore can again apply the recompression procedure thus further shrinking the length of $R$. The algorithm first again checks whether $|R| > 2^{j+\BBB} \cdot \frac{n}{\tau}$ and, if so, repeats the recompression. For the given fixed $j$, we do this iteratively until $|R| \le 2^{j+\BBB} \cdot \frac{n}{\tau}$. During this process, the number of zero values $M_i[j]$ in the arrays $M_i$ is always at most $2^j \cdot \frac{n}{\tau}$ since the equality $M_i[j] = 0$ implies that $S''' \cap (p..p{+}\tau/2^j] = \emptyset$, for a set $S'''\subseteq S'$ of size $|R|$ defined by analogy to the definition of $S''$ and for a position $p \in S'''$ corresponding to $i$. Therefore, due to Lemma~\ref{lem:recompression}, the condition $|R| \le 2^{j+\BBB} \cdot \frac{n}{\tau}$ eventually should be satisfied. Furthermore, as we are to show, for each $j$, the condition $|R| \le 2^{j+\BBB}\cdot \frac{n}{\tau}$ holds after at most three recompression iterations.

Given $j \in [\LB..\lceil\log^{(4)} n\rceil)$, the length of $R$ before the first iteration of the recompression for $j$ is at most $2^{j+\BB} \cdot \frac{n}{\tau}$ since this is a condition under which shrinking iterations stopped for $j + 1$. The same bound holds for $j = \lceil\log^{(4)} n\rceil$: the initial length of $R$ is at most $2^{\BB} \cdot\frac{n}{\tau} \log^{(3)} n$ (which is upper-bounded by $2^{j+\BB} \cdot\frac{n}{\tau}$) since $S' \subseteq S$ and $S$ is produced by the $k$th phase of the procedure of Section~\ref{sec:vishkin-process}, for $k = \lfloor\log\frac{\tau}{2^4\log^{(3)} n}\rfloor$, so that the size of $S$, by Lemma~\ref{lem:local-sparsity}, is at most $2^6\lceil n / 2^{k}\rceil \le 2^6 n / \frac{\tau}{2\cdot 2^4\log^{(3)} n} = 2^{11} \frac{n}{\tau} \log^{(3)} n$. Fix $j \in [\LB..\lceil\log^{(4)} n\rceil]$. Since the number of zero values $M_i[j]$ is always at most $2^j \cdot n / \tau$ and the number of zero values $M_{i-1}[5] = 0$ is at most $2^5\cdot \frac{n}{\tau}$, three iterations of the recompression for $j$ performed on a string $R$ with initial length $r$ shrink the length of $R$ to a length at most $(\frac{3}{4})^3 r + 2^j \cdot \frac{n}{\tau} + 2^5\cdot \frac{n}{\tau} \le (\frac{3}{4})^3 r + 2\cdot 2^{j} \cdot \frac{n}{\tau}$, by Lemma~\ref{lem:recompression}. Putting $r = 2^{j+\BB} \cdot \frac{n}{\tau}$, we estimate the length of $R$ after three iterations for $j$ from above by $((\frac{3}{4})^3 2^{\BB} + 2)2^j \cdot \frac{n}{\tau} < 2^{j+\BBB} \cdot \frac{n}{\tau}$. That is, for each $j$, three iterations are enough to reduce the length of $R$ to at most $2^{j+\BBB} \cdot \frac{n}{\tau}$.

Thus, the total running time of all recompression procedures is $\Oh(\sum_{j=\lceil\log^{(4)} n\rceil}^{\LB} 2^{j+\BB} \cdot\frac{n}{\tau} \log^{(4)} n) = \Oh(\frac{n}{\tau}\log^{(4)} n)$, which is $\Oh(n)$ since $\tau \ge (\log^{(3)} n)^4$. Observe that the most time consuming part is in recalculations of the arrays $M_i$, each taking $\Oh(|R|\log^{(4)} n)$ time, all other parts take $\Oh(|R|)$ time, i.e., $\Oh(\sum_{j=\lceil\log^{(4)} n\rceil}^{\LB} 2^{j+\BB} \cdot\frac{n}{\tau}) = \Oh(\frac{n}{\tau})$ time is needed for everything without the recalculations. The length of $R$ in the end is at most $2^{\LBplusBBB}\cdot n / \tau$, which is a condition under which shrinking iterations stopped for $j = \LB$.

Finally, we create a bit array $E$ of the same length as the original string $R$ that marks by $1$ those letters that survived all iterations. Additional navigational structures for linear-time $E$ construction are straightforward. We then re-run whole ``semi-online'' algorithm that generates the set $S'$ (from which the string $R$ was constructed) but, in this time, we discard all positions of $S'$ that correspond to unmarked indices in $E$ 
and we store all positions corresponding to marked indices of $E$ explicitly in an array $S^*$. Since at most $2^{\LBplusBBB}\cdot n / \tau$ indices in $E$ are marked by~$1$, the size of $S^*$ is $\Oh(n / \tau)$.

Finally, we have all required instruments to prove the main lemma. The proof is rather technical and, in a way, similar to the proof of Lemma~\ref{lem:2k-partitioning}; it is detailed in \appendixv{appx:main-lemma}{D.2}.
\begin{restatable}{lemma}{mainLemma}
The set $S^*$ is $\tau$-partitioning; also a converse of property~(c) holds for $S^*$: if a substring $s[i..j]$ has a period at most $\tau / 4$, then $S^* \cap [i + \tau .. j - \tau] = \emptyset$.\label{lem:main-lemma}
\end{restatable}



\bibliography{refs}

\ifdefined\ArxivVersion
\appendix
\newpage

\section{LCE index and sparse suffix tree}
\label{appx:sst-special}
\subsection{\boldmath Missing proofs}
\sstSpecial*
\begin{proof}
Denote by $j_k$ the $k$th position in $S$ (so that $j_1 < \cdots < j_{b}$). Assume that the letters $s[i]$ with $i \ge m$ are equal to the special letter $-1$; with this condition, substrings $s[j..j']$ with $j' \ge m$ are well defined. In one pass, we collect all substrings $s[j_k{+}d\tau .. j_k{+}d\tau{+}2\tau]$ with $j_k \in S$ and $d \in [0..(j_{k+1} - j_k) / \tau)$ into a set $C$ (each substring is identified by its starting position), defining $j_0 = 0, j_{b+1} = m$ to cover all of $s$ by the substrings. We order $C$ by starting positions from left to right: $C = \{s[i_h .. i_h{+}2\tau] \colon 1 \le h \le |C|\}$, where $i_1 < i_2 < \cdots < i_{|C|}$. Observe that $i_{h+1} - i_h \le \tau$, for any $h \in [1..|C|)$. Since the number of substrings in $C$ is $b + \Oh(m / \tau) = \Oh(b)$ and their total length is $\Oh(b\tau + m) = \Oh(m)$, they can be sorted in $\Oh(b)$ space using either the radix sort, which gives $\Oh(m\log_b \sigma)$ time, or the ternary tree~\cite{BentleySedgewick}, which leads to $\Oh(m + b \log b)$ time. We choose the fastest method of sorting and obtain $\Oh(m + \min\{m \log_b \sigma, b \log b\})$ time. Let $r_h$ be the rank of $s[i_h..i_h{+}2\tau]$ in the sorted order (equal strings have equal ranks).

We build a suffix tree $T'$ for the string $r_1 r_2 \cdots r_{|C|}$ in $\Oh(|C|) = \Oh(b)$ time and space~\cite{Farach}. All suffixes $r_h r_{h+1} \cdots r_{|C|}$ such that $i_h \not\in S$ are then removed from $T'$. By property~(b), the equality $r_h = r_{h'}$, for any $h, h' \in [1..|C|]$ such that $i_h \in S$ and $i_{h'} \in S$, implies that, for all $d \in [0..\tau]$, $i_h + d \in S$ iff $i_{h'} + d \in S$. Hence, $s[i_h .. i_h{+}2\tau] = s[i_{h'} .. i_{h'}{+}2\tau]$ and $i_{h+1} - i_h = i_{h'+1} - i_{h'}$ since $i_{h+1} - i_h \le \tau$, for $h \in [1..|C|)$. We inductively deduce from this that if $r_h r_{h+1} \cdots r_{h+\ell-1} = r_{h'} r_{h'+1} \cdots r_{h'+\ell-1}$, for $\ell \ge 1$ and $h, h' \in [1..|C|]$ such that $i_h \in S$ and $i_{h'} \in S$, then $s[i_h .. i_{h+\ell-1}{+}2\tau] = s[i_{h'} .. i_{h'+\ell-1}{+}2\tau]$ and, for all $d \in [0 .. i_{h+\ell-1} - i_h + \tau]$, $i_h + d \in S$ iff $i_{h'} + d \in S$. Therefore, if $s[i_h .. t{-}1] = s[i_{h'} .. t'{-}1]$ and $s[t] < s[t']$, for some $t, t'$, and $i_h, i_{h'} \in S$, then $r_h r_{h+1} \cdots r_{h+\ell-1} = r_{h'} r_{h'+1} \cdots r_{h'+\ell-1}$ and $r_{h+\ell} < r_{h'+\ell}$, where $\ell \ge 0$ is the smallest non-negative integer such that $t \in [i_{h+\ell} .. i_{h+\ell}{+}2\tau]$. The tree $T'$ is transformed into $T$ as follows. For each node in $T'$ that corresponds to a string $r_h r_{h+1} \cdots r_{h+\ell-1}$ and for each pair of its lexicographically adjacent outgoing edges labelled by $r_{h+\ell}$ and $r_{h'+\ell}$, we find in $\Oh(\tau)$ time the first mismatched positions $t$ and $t'$ in $s[i_{h+\ell}..i_{h+\ell}{+}2\tau]$ and $s[i_{h+\ell}..i_{h+\ell}{+}2\tau]$, and, using this information, create a corresponding node of $T$. The total running time of the transformation is $\Oh(|C|\tau) = \Oh(m)$.
\end{proof}

\section{Refinement of Partitioning Sets}
\subsection{\boldmath Missing proofs for locally consistent parsing}
\label{appx:lcp}



For the following proof, we need the well-known periodicity lemma.

\begin{lemma}[see~\cite{FineWilf}]
If two substrings $s[i..i']$ and $s[j..j']$ with periods $p$ and $q$, respectively, overlap on at least $p + q$ letters (i.e., $\min\{i',j'\} - \max\{i,j\} > p + q$), then the minimal period of their union, $s[\min\{i,j\}..\max\{i',j'\}]$, is at most $\mathop{\mathrm{gcd}}(p,q)$.\label{lem:fine-wilf}
\end{lemma}

\allR*
\begin{proof}
Local minima of $v_h$ are not in the region since, for $h \in [p..q]$, $R(j_h)$ implies $s_h = s_{h+1}$ and $v_h = \infty$. As for the ``boundary'' cases: (iii)~always gives $j_q \in S_k$; (i)~can affect only $j_p$ if $j_{p-1} - j_p > 2^{k-1}$; (ii)~holds for $j_p$ iff $q - p \ge 2$.

For any $h \in [p..q]$, since $s[j_h..j_h{+}2^k] = s[j_{h+1}..j_{h+1}{+}2^k]$ and $j_{h+1} - j_h \le 2^{k-1}$, the string $s[j_h..j_{h+1}{+}2^k]$ has period $j_{h+1} - j_h \le 2^{k-1}$. Then, for $h \in (p..q]$, $s[j_{h-1}..j_{h}{+}2^k]$ and $s[j_h..j_{h+1}{+}2^k]$ both have periods at most $2^{k-1}$ and overlap on $2^k$ letters. Hence, by Lemma~\ref{lem:fine-wilf}, their minimal periods coincide and are at most $2^{k-1}$. Applying this argument for all $h \in (p..q]$, we deduce that the minimal period of $s[j_p..j_q{+}2^k]$ is at most $2^{k-1}$.
\end{proof}

\allNonR*
\begin{proof}
We have either $v'_q = \infty$ (if $j_{q+1} - j_q > 2^{k-1}$) or $v'_{q+1} = \infty$ (if $R(j_{q+1})$ holds). Since the numbers $v_h$ are obtained via four $\vbit$ reductions and, for any $x$, $\vbit(x,\infty) = \infty$, we have $v_h = \infty$, for $h \in (q{-}3..q]$, and also $v_{q-3} = \infty$ if $v'_q = \infty$. Since $s_h \ne s_{h+1}$ provided $h \in [p..q)$, the claim for other $v_h$ with $h \in [p..q{-}4]$ is proved by four applications of Lemma~\ref{lem:vishkin} to the sequence $s_{p}, s_{p+1}, \ldots, s_q$. The criteria for $j_h \in S_k$ follow directly from the refinement rule of the $k$th phase.
\end{proof}

\subsection{\boldmath Proof that the refinement process builds a hierarchy of partitioning sets}

\localSparsity*
\begin{proof}
\label{appx:refinement-sparsity}
Let us prove by induction on $k$ the following claim: the range $[0..n)$ of positions of the string $s$ can be partitioned into disjoint non-empty blocks $B_1, B_2, \ldots, B_m$ such that $B_i = [b_i..b_{i+1})$, where $b_1 = 0$, $b_{m+1} = n$, and $b_1 < b_2 < \cdots < b_{m+1}$, and the blocks are of two types:
\begin{enumerate}
\item a block $B_i$ is \emph{normal} if $2^{k-5} \le |B_i| \le 4\cdot 2^{k-5}$ and $|B_i \cap S_k| \le 2$;
\item a block $B_i$ is \emph{skewed} if $|B_i| \ge 8\cdot 2^{k-5}$, $|B_i \cap S_k| \le 3$, and $[b_i..b_{i+1}{-}4\cdot 2^{k-5}] \cap S_k = \emptyset$, i.e., all positions of $B_i \cap S_k$ are concentrated in the suffix of $B_i$ of length $4\cdot 2^{k-5}$ (hence the name ``skewed'').
\end{enumerate}
The claim implies the lemma as follows. The worst case, which maximizes $|S_k \cap [i..i')|$, occurs when a range $[i..i')$ intersects one skewed block and at most $\lceil\frac{i'-i-3}{2^{k-5}}\rceil$ normal blocks in such a way that the skewed block is the leftmost one and contains three positions from $S_k$, all inside $[i..i')$, and each normal block contains two positions from $S_k$, all inside $[i..i')$. We then have $|S_k \cap [i..i')| \le 2\lceil\frac{i'-i-3}{2^{k-5}}\rceil + 3$. The bound can be rewritten as $2\lceil \frac{c 2^k + d}{2^{k-5}}\rceil + 3$, substituting integers $c \ge 0$ and $d \in [0..2^k)$ such that $i' - i - 3 = c 2^k + d$, and it is upper-bounded by $2^6c + \lceil \frac{d}{2^{k-5}}\rceil + 3 \le 2^6 c + 2^5 + 3 < 2^6(c + 1) = 2^6 \lceil\frac{i' - i - 3}{2^k}\rceil \le 2^6  \lceil\frac{i' - i}{2^k}\rceil$, as the lemma states.


Now let us return to the inductive proof of the claim.
The base of the induction is $k \le 6$ and it is trivial since the range $[0..n)$ always can be split into blocks of length $2$. Suppose that the inductive hypothesis holds for $k \ge 6$ and let us construct a partitioning of the range $[0..n)$ into blocks for $k + 1$.

First, we greedily unite blocks into disjoint pairs and singletons from left to right as follows: given a block $B_i$ such that the blocks $B_1 B_2 \cdots B_{i-1}$ were already united into pairs and singletons, we consider the following cases: (1)~if $B_i$ is a skewed block, then it forms a singleton and we analyze $B_{i+1}$ next; (2)~if both $B_i$ and $B_{i+1}$ are normal blocks, then $B_i$ and $B_{i+1}$ are united into a pair and we consider $B_{i+2}$ next; (3)~if $B_i$ is a normal block and $B_{i+1}$ is a skewed block, then we cut from $B_{i+1}$ a new normal block $B' = [b_{i+1}..b_{i+1}{+}2^{k-5})$ of length $2^{k-5}$, which is necessarily empty due to properties of skewed blocks, and we unite $B_i$ with $B'$ and analyze the (already cut) block $B_{i+1}$ next; the length of the skewed block $B_{i+1}$ is reduced by $2^{k-5}$ after the cutting and might become less than $8\cdot 2^{k-5}$ but it is still considered as skewed since this is not an issue as the block $B_{i+1}$ will be anyways dissolved into normal blocks in a moment. After the construction of all the pairs and singletons, we proceed as follows.

We consider all the produced pairs and singletons of blocks from left to right. A singleton block $B_i$ is always a skewed block and its length is at least $7\cdot 2^{k-5}$ (a prefix of length $2^{k-5}$ could be cut from the block). Let $S_k = \{j_1 < \cdots < j_{|S_k|}\}$. It follows from the refinement rule that any two consecutive positions $j_h$ and $j_{h+1}$ from $S_{k}$ can both belong to the set $S_{k+1}$ only if either $j_{h+1} - j_h > 2^{k}$, or $j_{h} - j_{h-1} > 2^{k}$ (and $v_{h+1}$ is a local minimum that is not preceded by $\infty$ in the latter case). In both cases there is a gap of length at least $2^{k}$ between either $j_h$ and $j_{h+1}$, or $j_{h-1}$ and $j_{h}$. Therefore, since only the last $4\cdot 2^{k-5} = 2^{k-3}$ positions of the skewed block $B_i$ may contain positions of $S_k$, the refinement rule in the case $|B_i \cap S_k| = 3$ necessarily removes from $S_k$ either the second or the last position of $S_k \cap B_i$ when producing $S_{k+1}$. Thus, we have $|B_i \cap S_{k+1}| \le 2$. We then split the skewed block $B_i$ into a series of new normal blocks all of which, except possibly the last one, do not contain positions from $S_{k+1}$: the last block is $(b_{i+1}{-}4\cdot 2^{k-5}..b_{i+1}]$ and it has length $4\cdot 2^{k-5} = 2\cdot 2^{k-4}$; the remaining prefix of $B_i$ is $[b_i..b_{i+1}{-}4\cdot 2^{k-5}]$ (the boundary $b_i$ used here can differ from the original $b_i$ if the skewed block $B_i$ was cut and, thus, $b_i$ was increased) and its length is at least $7\cdot 2^{k-5} - 4\cdot 2^{k-5}= 3\cdot 2^{k-5}$, the prefix is split into normal blocks arbitrarily (recall that the length of a normal block for the inductive step $k+1$ must be at least $2^{k-4}$ and at most $4\cdot 2^{k-4}$).

Consider a pair of normal blocks $B_i$ and $B_{i+1}$. For simplicity, we denote by $B_{i+1}$ the block following $B_i$ even if it was created by cutting the skewed block (actual $B_{i+1}$) that followed $B_i$ in the initial partitioning. The length of the united block $B_i B_{i+1}$ is at least $2^{k-4}$ and at most $4\cdot 2^{k-4}$ so that it could have formed a new normal block if at most two positions of $S_{k+1}$ belonged to it. By the observation above, if two consecutive positions $j_h$ and $j_{h+1}$ were retained in $S_{k+1}$ by the refinement rule, then there is a gap of length at least $2^{k}$ between either $j_h$ and $j_{h+1}$, or $j_{h-1}$ and $j_{h}$. Therefore, since $|B_iB_{i+1}| \le 4\cdot 2^{k-4} < 2^{k}$, the united block $B_i B_{i+1}$ contains at most three positions of $S_{k+1}$, i.e., one of the positions from $S_k \cap B_iB_{i+1}$ must be discarded from $S_{k+1}$. In case $|S_{k+1} \cap B_iB_{i+1}| \le 2$, we simply form a new normal block $B_iB_{i+1}$. The case $|S_{k+1} \cap B_iB_{i+1}| = 3$ is more interesting; it occurs when each of the blocks $B_i$ and $B_{i+1}$ contains two positions from $S_k$ and the refinement procedure retains the two positions from $S_k \cap B_i$ and removes the first position of $S_k \cap B_{i+1}$. By the observation above, we must have a gap before the block $B_i$ in this case: $S_{k+1} \cap [b_i{-}2^{k}{+}|B_i|..b_i) = \emptyset$. Since $|B_i| \le 4\cdot 2^{k-5} = 2^{k-3}$, we hence obtain $S_{k+1} \cap [b_i{-}7\cdot 2^{k-3}..b_i) = \emptyset$. Therefore, all newly produced blocks that are entirely contained in the range $[b_i{-}7\cdot 2^{k-3}..b_i)$ are empty, i.e., they do not contain positions of $S_k$ (recall that we process pairs and singletons of blocks from left to right and, so, only newly constructed blocks are located to the left of $B_i$). Let $\hat{B}_j, \hat{B}_{j+1}, \ldots, \hat{B}_{\ell}$ be a maximal sequence of consecutive newly constructed empty blocks to the left of $B_i$ such that $\hat{B}_{\ell}$ is a block immediately preceding $B_i$ (we use the notation $\hat{B}_j$ for the new blocks to distinguish them from the ``old'' blocks $B_1, B_2, \ldots$). We unite all blocks $\hat{B}_j, \hat{B}_{j+1}, \ldots, \hat{B}_{\ell}$  with $B_i$ and $B_{i+1}$ thus producing a new skewed block whose length is at least $|B_i B_{i+1}| + 7\cdot 2^{k-3} - 4\cdot 2^{k-4}$ (the negative term is because some block to the left of $B_i$ may only partially intersect the ``empty'' range $[b_i{-}7\cdot 2^{k-3}..b_i)$), which is at least $2^{k-4} + 5\cdot 2^{k-3} = 11\cdot 2^{k-4} \ge 8\cdot 2^{k-4}$.
\end{proof}

\partitioningTau*
\begin{proof}
\label{appx:refinement-partitioning}
The proof is by induction on $k$. The case $k = 0$ is trivial. Let $k > 0$ and $S_{k-1} = \{j_1 < \cdots < j_{|S_{k-1}|}\}$. We first establish property~(a) for $S_k$ with $\tau = 2^{k+3}$: if $s[i{-}2^{k+3}..i{+}2^{k+3}] = s[j{-}2^{k+3}..j{+}2^{k+3}]$ for $i,j \in [2^{k+3}..n{-}2^{k+3})$, then $i \in S$ iff $j \in S$. Fix $i$ and $j$ such that $s[i{-}2^{k+3}..i{+}2^{k+3}] = s[j{-}2^{k+3}..j{+}2^{k+3}]$. It suffices to show that $i \in S_k$ implies $j \in S_k$; the case $j \in S_k$ is symmetric. So, let $i \in S_k$.

By the inductive hypothesis, for any $\ell \in [-2^{k+2}..2^{k+2}]$, we have $i + \ell \in S_{k-1}$ iff $j + \ell \in S_{k-1}$. In particular, $i$ and $j$ both belong to $S_{k-1}$. Let $i = j_h$ and $j = j_{h'}$, for some $h$ and $h'$. Therefore, when $i \in S_k$ due to the ``boundary'' case~(i), i.e., either $j_h - j_{h-1} > 2^{k-1}$ or $j_{h+1} - j_{h} > 2^{k-1}$, then $j_{h'} - j_{h'-1} > 2^{k-1}$ or $j_{h'+1} - j_{h'} > 2^{k-1}$, respectively, as $2^{k-1} < 2^{k+2}$; thus, $j \in S_k$.

Now suppose that $i \in S_k$ because $\infty > v_{h-1} > v_h$ and $v_h < v_{h+1}$, where $v_{h-1}, v_h, v_{h+1}$ are as defined in the transformation of $S_{k-1}$ into $S_k$. By Lemma~\ref{lem:all-non-R}, whenever $v_h \ne \infty$, then $j_{h+\ell} - j_{h+\ell-1} \le 2^{k-1}$, for all $\ell \in [1..4]$, and the value of $v_h$ is derived using the substrings $s[j_{h+\ell}..j_{h+\ell}{+}2^k]$ with $\ell \in [0..4]$. Then, $j_{h+\ell} - j_{h+\ell-1} \le 2^{k-1}$, for all $\ell \in [0..4]$. Since $j_{h+4} - j_h \le 4\cdot 2^{k-1} < 2^{k+2}$, it follows from the inductive hypothesis that $j_{h+\ell} - j_{h+\ell-1} = j_{h'+\ell} - j_{h'+\ell-1}$, for all $\ell \in [0..4]$, and $s[j_{h+\ell} .. j_{h+\ell}{+}2^k] = s[j_{h'+\ell} .. j_{h'+\ell}{+}2^k]$, for all $\ell \in [-1..4]$. Hence, $v_{h'-1} = v_{h-1}$ and $v_{h'} = v_{h}$. In the same fashion, the case $v_{h+1} \ne \infty$ implies that $j_{h+5} - j_h \le 5\cdot 2^{k-1} < 2^{k+2}$ and $s[j_{h+5} .. j_{h+5}{+}2^k] = s[j_{h'+5} .. j_{h'+5}{+}2^k]$, thus inferring $v_{h'+1} = v_{h+1}$. By Lemma~\ref{lem:all-non-R}, in case $v_{h+1} = \infty$ either $j_{h+5} - j_{h+4} > 2^{k-1}$, which gives $j_{h'+5} - j_{h'+4} > 2^{k-1}$ since $j_{h+4} + 2^{k-1} - j_h < 2^{k+2}$, or $s[j_{h+4} .. j_{h+4}{+}2^k] = s[j_{h+5} .. j_{h+5}{+}2^k]$ and $j_{h+5} - j_{h+4} \le 2^{k-1}$, which gives $s[j_{h'+4} .. j_{h'+4}{+}2^k] = s[j_{h'+5} .. j_{h'+5}{+}2^k]$; thus, both alternatives imply $v_{h'+1} = \infty$. Therefore, $j$ (${=}j_{h'}$) belongs to $S_k$ because $v_{h'}$ is a local minimum. Analogously, we deduce that when $i \in S_k$ since $R(j_{h-1})$ does not hold but $R(j_h)$, $R(j_{h+1})$, $R(j_{h+2})$ hold, then the same is true for $R(j_{h'-1})$ and $R(j_{h'})$, $R(j_{h'+1})$, $R(j_{h'+2})$, respectively, and, thus, $j \in S_k$. The remaining case when $i \in S$ since $R(j_h)$ holds but $R(j_{h+1})$ does not hold is similar.

Let us establish property~(b) for $S_k$ with $\tau = 2^{k+3}$: if $s[i..i{+}\ell] = s[j..j{+}\ell]$, for $i,j \in S_k$ and some $\ell \ge 0$, then, for each $d \in [0..\ell{-}2^{k+3})$, $i + d \in S_k$ iff $j + d \in S_k$. In view of property~(a), the only interesting case is when $\ell > 2^{k+3}$ and $d \in (0..2^{k+3})$. Given $i + d \in S_k$, let us show that $j + d \in S_k$; the case $j + d \in S_k$ is symmetric. By the inductive hypothesis, we have $i + d' \in S_{k-1}$ iff $j + d' \in S_{k-1}$, for any $d' \in [0..\ell{-}2^{k+2})$; moreover, $s[i{+}d' .. i{+}d'{+}2^k] = s[j{+}d' .. j{+}d'{+}2^k]$, for such $d'$, since $2^k < 2^{k+2}$. In particular, $i + d$ and $j + d$ both belong to $S_{k-1}$. Let $i + d = j_h$ and $j + d = j_{h'}$, for some $h$ and $h'$. The remaining case analysis is exactly as in the proof of property~(a): the only strings of interest to the ``left'' of $i + d$ and $j + d$ in the analysis are $s[j_{h-1} .. j_{h-1}{+}2^k]$ and $s[j_{h'-1} .. j_{h'-1}{+}2^k]$, respectively, and they coincide since $j_{h} - j_{h-1} = j_{h'} - j_{h'-1} \le d$; all strings to the ``right'' are addressed using the inductive hypothesis.

It remains to establish property~(c) for $S_k$ with $\tau = 2^{k+3}\lfloor\log^{(3)} n\rfloor$: if $i,j \in S_k$ with $i < j$, $(i..j) \cap S_k = \emptyset$, and $j - i > \tau$, then $s[i..j]$ has a period at most $\tau / 4$. We shall actually prove a stronger claim that $s[i..j]$ has a period at most $2^{k-1}$. We assume that $n > 2^{16}$ so that $\tau > 2^{k+3}$. If $(i..j) \cap S_{k-1} = \emptyset$, then the property follows from the inductive hypothesis. Assume that the set $(i..j) \cap S_{k-1}$ is not empty and denote by $j_r, j_{r+1}, \ldots, j_t$ all its positions. Note that $j_\ell \not\in S_k$, for all $\ell \in [r..t]$, and $i = j_{r-1}$ and $j = j_{t+1}$. We have $j_{\ell} - j_{\ell-1} \le 2^{k-1}$, for all $\ell \in [r..t{+}1]$, since otherwise $j_{\ell-1}, j_\ell \in S_k$. The latter implies that $2 + t - r \ge (j - i) / 2^{k-1} > \tau / 2^{k-1} = 2^4\lfloor\log^{(3)} n\rfloor$.

Suppose that $R(j_{u})$ holds, for some $u \in [r..t]$, and let $u$ be the smallest such $u$. By Lemma~\ref{lem:all-R}, the last position of the all-$R$ region containing $j_u$ belongs to $S_k$ and its other positions, except possibly the first one, do not belong to $S_k$. Hence, $R(j_\ell)$ holds, for all $\ell \in [u..t{+}1]$, and $j$ ($=j_{t+1}$) is this last position. Suppose that $t + 1 - u \ge 2$. Then, by Lemma~\ref{lem:all-R}, the first position in the all-$R$ region containing $j_u$ belongs to $S_k$ too. Hence, this first position must be $j_{r-1} = i$. Therefore, the string $s[i..j]$ has a period at most $2^{k-1}$ due to Lemma~\ref{lem:all-R}.

Suppose that $R(j_r)$ does not hold, i.e., $j_r$ belongs to an all-non-$R$ region according to Lemma~\ref{lem:all-non-R}. Then, by the above argument, $R(j_u)$ may hold only for $u = t$ in $[r..t]$ (so that $ t + 1 - u < 2$) and, therefore, all $j_r, j_{r+1}, \ldots, j_{t-1}$ are in the same all-non-$R$ region. By Lemma~\ref{lem:all-non-R}, we have $v_\ell \ne \infty$ when $\ell \in [r..t{-}5]$. Hence, $(t - 5) - r + 1 < 8\log^{(3)} n + 12$ since otherwise, by Lemma~\ref{lem:local-density}, $j_\ell \in S_k$, for some $\ell \in [r..t{-}5]$. But it contradicts the inequality $t - r > 2^4\lfloor\log^{(3)} n\rfloor - 2$ derived above, assuming $n$ is large enough. Thus, this case (when $R(j_r)$ does not hold) is impossible.
\end{proof}

\section{Speeding up the Refinement Procedure}
\label{appx:refinement-speed}
\subparagraph{The case \boldmath $\tau < \sqrt{n}$.}
Each phase in this scheme uses $\Oh(1)$ space so that the overall space is $\Oh(\log\tau)$, which, as we assumed, is always available. Suppose that $\tau < \sqrt{n}$ . We have $b = \Theta(\frac{n}{\tau}) \ge \Omega(\sqrt{n})$ additional space for the algorithm in this case, which can be utilized in order to speed up LCE queries, the bottleneck of the scheme, as follows. Consider the (overlapping) substrings $C_i = s[i\lfloor\sqrt{n}\rfloor..(i + 3)\lfloor\sqrt{n}\rfloor - 1]$, for $i \in [0..n / \lfloor\sqrt{n}\rfloor{-}3]$. We build in linear time for $C_0$ an LCE data structure~\cite{HarelTarjan} using $\Oh(|C_0|) = \Oh(\sqrt{n})$ space that can answer LCE queries in $\Oh(1)$ time. Then, the algorithm processes the prefix $s[0.. 2\lfloor\sqrt{n}\rfloor{-}1]$ from left to right feeding the positions $[0.. 2\lfloor\sqrt{n}\rfloor)$ to the first and all subsequent phases and computing queries $\min\{2^k + 1, \lce(j, j')\}$, with $j, j' \in [0.. 2\lfloor\sqrt{n}\rfloor)$, emerging in the phases along the way in $\Oh(1)$ time; the latter is possible since $2^k \le \frac{\tau}{2^4\log^{(3)} n} < \frac{1}{8}\sqrt{n}$ for all $k \le \lfloor\log\frac{\tau}{2^4\log^{(3)} n}\rfloor$, assuming $n > 2^{16}$, and hence, the strings $s[j..j{+}2^k]$ and $s[j'..j'{+}2^k]$ in the queries are substrings of $C_0$. Since $6\cdot 2^k < \sqrt{n}$ for $k \le \lfloor\log\frac{\tau}{2^4\log^{(3)} n}\rfloor$, it follows from Lemma~\ref{lem:dist-process} that any subsequent LCE queries $\lce(j, j')$ in the algorithm that will emerge after the processing of the prefix $s[0.. 2\lfloor\sqrt{n}\rfloor{-}1]$ can be performed only on positions $j$ and $j'$ such that $\min\{j,j'\} \ge \sqrt{n}$, i.e., the positions $j$ and $j'$ and the corresponding substrings $s[j..j{+}2^k]$ and $s[j'..j'{+}2^k]$ in the queries will be inside the string $C_1$ in the next $\sqrt{n}$ steps. Accordingly, to continue the execution of the algorithm, we build an LCE data structure for $C_1$ in place of the structure for $C_0$ and continue the run feeding the positions $[2\lfloor\sqrt{n}\rfloor..3\lfloor\sqrt{n}\rfloor)$ to the first and all subsequent phases. We continue this procedure analogously: on a generic step, after feeding the positions $[i\lfloor\sqrt{n}\rfloor..(i+1)\lfloor\sqrt{n}\rfloor)$ using an LCE data structure for $C_{i-1}$, we construct an LCE data structure for $C_i$ in its place in $\Oh(|C_i|)$ time and feed the positions $[(i+1)\lfloor\sqrt{n}\rfloor..(i+2)\lfloor\sqrt{n}\rfloor)$ to the algorithm. The overall running time is $\Oh(n + \sum_i |C_i|) = \Oh(n)$ and the occupied space is $\Oh(\sqrt{n}) = \Oh(b)$.

\subparagraph{The case \boldmath $\tau \ge \sqrt{n}$.}
Let us generalize this idea to the case $\tau \ge \sqrt{n}$. Denote $b = \frac{n}{\tau}$. We have $\Oh(b) < \Oh(\sqrt{n})$ space and, hence, cannot resort to the above described scheme since LCE data structures for substrings of length $\Oh(b)$ are not enough to answer queries of the form $\min\{2^k + 1, \lce(j, j')\}$ when $2^k > \Omega(b)$. The key idea is that the queries $\lce(j, j')$ in a $k'$th phase can use only positions $j$ and $j'$ such that $j, j' \in S_{k}$, for all $k \in [0..k')$, and therefore, the full suffix tree that lies at the core of the LCE data structure~\cite{HarelTarjan} used for the substrings $C_i$ is unnecessary, we can use a sparse suffix tree only for suffixes whose starting positions are from a set $S_{k}$, for some $k < k'$. When the sparse suffix tree is equipped with an LCA data structure~\cite{HarelTarjan}, one can compute the queries $\lce(j,j')$, for $j, j' \in S_{k'-1}$, in $\Oh(1)$ time.

Denote $\hat{b} = \lfloor\frac{b}{\log n}\rfloor$. Our new scheme evenly splits all $\lfloor\log\frac{\tau}{2^4\log^{(3)} n}\rfloor$ phases into ``levels'', each containing $\Theta(\log\hat{b})$ phases. Observe that, since $b = \Omega(\log^2 n)$, we have $\log\hat{b} = \Theta(\log b)$. The scheme has $\lfloor\log(\frac{\tau}{2^4\log^{(3)} n}) / \log\hat{b}\rfloor = \Oh(\log_{\hat{b}} n)$ ``levels'', which is $\Oh(\log_b n)$ as $\log\hat{b} = \Theta(\log b)$. The idea is that each level has its own sparse suffix tree, endowed with an LCA data structure, that occupies $\Oh(\hat{b})$ space and can answer in $\Oh(1)$ time LCE queries on a range of $\Theta(\hat{b})$ positions from a set $S_{k}$ received by the lowest phase in the level; all phases inside the level use this tree for LCE queries, which is possible since the phases process sets $S_{k+1}, S_{k+2}, \ldots$ and $S_k \supseteq S_{k+1} \supseteq S_{k+2} \supseteq \cdots$. As in the solution for $\tau < \sqrt{n}$, the sparse suffix tree encodes only a certain range of the input at any given time and it has to be rebuild, in a ``sliding window'' manner, once a long enough prefix of the input has been processed. The total time required for all rebuilds in one level is $\Oh(n\log_b n)$, for the lowest level, and $\Oh(n)$, for all other levels. Thus, the time over all $\Theta(\log_b n)$ levels is $\Oh(n\log_b n)$. The overall space is $\Oh(\hat{b}\log_b n) = \Oh(b)$ (this is why the parameter $\hat{b}$ was introduced: in order to bound the space on each level separately). The precise description follows.

For integer $p \ge 0$, the $p$th level takes positions from the set $S_k$ produced by the previous level, for $k = p\lfloor\log\hat{b}\rfloor$, and produces the set $S_{k'}$, for $k' = (p + 1)\lfloor\log\hat{b}\rfloor$. The $p$th level receives positions of the set $S_{k}$ from left to right; for $p = 0$, the received positions are $0, 1, \ldots, n{-}1$. When sufficiently many positions of $S_{k}$ are collected (at most $\Oh(\hat{b})$), we temporarily pause the procedure of the previous level, process the collected chunk of positions thus producing some positions of $S_{k'}$ (from left to right), and then continue collecting positions of $S_{k}$ until the next chunk is ready, which is again processed analogously, and so on. Let us describe the scheme in details.

For integer $p \ge 0$, fix $k = p\lfloor\log\hat{b}\rfloor$ and $k' = (p + 1)\lfloor\log\hat{b}\rfloor$. The $p$th level works as follows.
By analogy to the case $\tau < \sqrt{n}$, we consider consecutively substrings $C_i = s[i\hat{b}\cdot 2^{k+3}..(i + 3)\hat{b}\cdot 2^{k+3} - 1]$, for $i \in [0..n / (\hat{b}\cdot 2^{k+3}) - 3]$. First, we collect all positions of $S_k$ produced by the $(p{-}1)$th level that are less than $3\hat{b}\cdot 2^{k+3}$ (i.e., span the substring $C_0$) and, then, temporarily pause the generation of new positions from $S_k$ (i.e., pause the levels $p-1, p-2, \ldots, 0$). Denote $Q_i = S_k \cap [i\hat{b}\cdot 2^{k+3}..(i + 3)\hat{b}\cdot 2^{k+3})$. By Lemma~\ref{lem:local-sparsity}, we have $|Q_i| \le \Oh(\hat{b})$. Since, by Lemma~\ref{lem:2k-partitioning}, the set $S_k$ is ``almost'' $2^k$-partitioning (it satisfies properties~(a) and~(b)), we can construct on the string $C_0$ a sparse suffix tree, for all its suffixes $s[j..|C_0|{-}1]$ starting at positions $j \in Q_0$, in $\Oh(|C_0| + \min\{|C_0|\log_{\hat{b}} n, \hat{b}\log\hat{b}\})$ time using Lemma~\ref{lem:sst-special}. This time can be upper-bounded by $\Oh(|C_0| \log_b n)$ for $p = 0$ (note that $\log_{\hat{b}} n = \Theta(\log_b n)$ since $b \ge \Omega(\log^2 n)$), and by $\Oh(|C_0| + \frac{|C_0|}{2^k}\log \hat{b}) = \Oh(|C_0| + \frac{|C_0|}{\hat{b}^p}\log\hat{b}) = \Oh(|C_0|)$ for $p \ge 1$ (recall that $k = p\lfloor\log\hat{b}\rfloor$). The sparse suffix tree is equipped in $\Oh(\hat{b})$ time by an LCA data structure~\cite{HarelTarjan} that allows us to compute LCE queries for substrings of $C_0$ starting at positions from $Q_0$ in $\Oh(1)$ time.

We build the set $S_{k'}$ as follows: we first process the positions $S_k \cap [0..2\hat{b}\cdot 2^{k+3})$ by the usual procedure of the phases $k + 1, k + 2, \ldots, k'$ using the sparse suffix tree of $C_0$ in order to answer LCE queries of the form $\min\{2^{k'} + 1, \lce(j,j')\}$ emerging along the way in $\Oh(1)$ time. This is possible since $2^{k'} \le \hat{b}\cdot 2^k$ and the queries $\lce(j,j')$ access only positions $j, j' \in Q_0 \cap [0..2\hat{b}\cdot 2^{k+3})$ so that $s[j..j{+}2^{k'}]$ and $s[j'..j'{+}2^{k'}]$ are substrings of $C_0$. By Lemma~\ref{lem:dist-process}, we shall report in this way all positions $S_{k'} \cap [0..2\hat{b}\cdot 2^{k+3} - 5\cdot 2^{k'}] \supseteq S_{k'} \cap [0 .. (1 + \frac{3}{8})\hat{b}\cdot 2^{k+3}]$ (note that since $2^{k'} \le \hat{b}\cdot 2^k$, we have $2\hat{b}\cdot 2^{k+3} - 5\cdot 2^{k'} \ge 2\hat{b}\cdot 2^{k+3} - 5\hat{b}\cdot 2^k = (1 + \frac{3}{8})2^{k+3}\hat{b}$). Then, we resume the procedures of the previous temporarily paused levels $p - 1, p-2, \ldots, 0$ that generate $S_k$ and collect all positions from $S_k \cap [3\hat{b}\cdot 2^{k+3}..4\hat{b}\cdot 2^{k+3})$. Once they are collected and, therefore, the set $Q_1$ is known, the previous levels $p - 1, p-2, \ldots, 0$ are again paused. We construct the sparse suffix tree of $C_1$, for all suffixes starting at positions from $Q_1$, in place of the tree for $C_0$ using Lemma~\ref{lem:sst-special}, which takes $\Oh(|C_1|\log_b n)$ or $\Oh(|C_1|)$ time depending on whether $p = 0$ or not. Further, we feed the positions $S_k \cap [2\hat{b}\cdot 2^{k+3}..3\hat{b}\cdot 2^{k+3})$ to the phases $k + 1, k + 2, \ldots, k'$ using the suffix tree to answer LCE queries. This is possible since, by Lemma~\ref{lem:dist-process}, all substrings $s[j..j{+}2^{k'}]$ that are accessed by LCE queries at this point have $j \ge 2\hat{b}\cdot 2^{k+3} - 6\cdot 2^{k'} > \hat{b}\cdot 2^{k+3}$. Thus, we report the rest of the set $S_{k'} \cap [0..3\hat{b}\cdot 2^{k+3} - 5\cdot 2^{k'}] \supseteq S_k \cap [0 .. (2 + \frac{3}{8})\hat{b}\cdot 2^{k+3}]$ as a result. Analogously, we continue the described procedure for $C_2, C_3, \ldots$: for $C_i$, we collect positions from $S_k \cap [(i+2)\hat{b}\cdot 2^{k+3}..(i+3)\hat{b}\cdot 2^{k+3})$, thus constructing $Q_i$, then temporarily pause the previous levels, construct the sparse suffix tree for the suffixes of $C_i$ starting at positions from $Q_i$, and feed the positions $S_k \cap [(i + 1)\hat{b}\cdot 2^{k+3}..(i + 2)\hat{b}\cdot 2^{k+3})$ to the phases $k + 1, k + 2, \ldots, k'$, thus reporting $S_{k'} \cap [0 .. (i + 1 + \frac{3}{8})\hat{b}\cdot 2^{k+3}]$ in the end; then, we move on to $C_{i+1}$.

Each level uses $\Oh(\hat{b})$ space. Since the number of levels is $\Oh(\log_b n)$, the total space is $\Oh(\hat{b}\log_b n) = \Oh(b)$. The overall running time of the first level ($p=0$) is $\Oh(\sum_i |C_i|\log_b n) = \Oh(n\log_b n)$, the time for each subsequent level ($p > 0$) is $\Oh(\sum_i |C_i|) = \Oh(n)$. Thus, the total time is $\Oh(n \log_b n)$.

\section{Recompression}
\subsection{\boldmath Proof that $S'$ obtained after the first stage is $\tau$-partitioning}
\label{appx:first-stage}
\periodicityGap*
\begin{proof}
For property~(a) of $S'$ (with $\frac{3}{4}\tau$ in place of $\tau$), consider $p$ and $q$ such that $s[p{-}\frac{3}{4}\tau .. p{+}\frac{3}{4}\tau] = s[q{-}\frac{3}{4}\tau..q{+}\frac{3}{4}\tau]$. It suffices to show that if $p \in S'$, then $q \in S'$. Since $S$ is a \mbox{$\frac{\tau}2$-partitioning} set and $S' \subseteq S$, the inclusion $p \in S'$ implies $q \in S$, by property~(a) of $S$ (with $\frac{\tau}{2}$ in place of $\tau$). The position $q$ could be excluded from $S'$ only if there exist $r \in (0..\tau/4]$ and $r' \in [0..\tau/4]$ such that $r' + r \le \tau/4$, $q - r \in S$, $q + r' \in S$, and $s[q{-}r..q{-}r{+}\frac{\tau}{2}] = s[q{+}r'..q{+}r'{+}\frac{\tau}{2}]$. Then, since $p - \frac{3}{4}\tau \le p - r - \frac{\tau}{2}$ and $p + r' + \frac{\tau}{2} \le p + \frac{3}{4}\tau$, we have $p - r \in S$ and $p + r' \in S$ by property~(a) of $S$ (with $\frac{\tau}{2}$ in place of $\tau$). Therefore, $p$ must be excluded from $S'$ too after the comparison of the equal substrings $s[p{-}r..p{-}r{+}\frac{\tau}{2}]$ and $s[p{+}r'..p{+}r'{+}\frac{\tau}{2}]$, which contradicts the assumption $p \in S'$.

For property~(b) of $S'$ (with $\frac{3}{4}\tau$ in place of $\tau$), consider $p,q \in S'$ such that $s[p..p{+}\ell] = s[q..q{+}\ell]$, for some $\ell \ge 0$. By contradiction, assume that there exists $d \in [0..\ell{-}\frac{3}{4}\tau)$ such that $p + d \in S'$ whereas $q + d \not\in S'$, and $d$ is the smallest such integer. Denote $p' = p + d$ and $q' = q + d$. Since $p' \in S' \subseteq S$, we have $q' \in S$, by property~(b) of $S$ (with $\frac{\tau}{2}$ in place of $\tau$). Then, as above, $q'$ could be excluded from $S'$ only if there are $r \in (0..\tau/4]$ and $r' \in [0..\tau/4]$ such that $r' + r \le \tau/4$, $q' - r \in S$, $q' + r' \in S$, and $s[q'{-}r..q'{-}r{+}\frac{\tau}{2}] = s[q'{+}r'..q'{+}r'{+}\frac{\tau}{2}]$. Hence, all positions $S \cap (q'{-}r .. q'{+}r']$ were excluded from $S'$, i.e., $S' \cap (q'{-}r .. q'{+}r'] = \emptyset$. Since $q \in S'$, we have $q' - r \ge q$. Since $d + r < \ell - \frac{\tau}{2}$ and $q' - r, q' + r' \in S$, we obtain $p' - r \in S$ and $p' + r \in S$, by property~(b) of $S$ (with $\frac{\tau}{2}$ in place of $\tau$). Therefore, the positions $S \cap (p'{-}r .. p'{+}r]$ all too should have been excluded from $S'$ together with $p'$, which is a contradiction.

For property~(c) of $S'$ (with $\tau$), consider $p, q \in S'$ such that $(p..q) \cap S' = \emptyset$ and $q - p > \frac{3}{4}\tau$ (the property requires to consider only $q - p > \tau$ but the claim works for $q  - p > \frac{3}{4}\tau$ too). We are to show that $s[p..q]$ has a period at most $\tau / 4$ (unfortunately, property~(c) for $\frac{3}{4}\tau$ in place of $\tau$ does not hold because, in this case, $s[p..q]$ should have a period at most $\frac{3}{16}\tau$, which is not true in general). If $(p..q) \cap S = \emptyset$, then the string $s[p..q]$ has a period at most $\tau / 8$ (${<}\tau/4$) because of property~(c) of $S$ (with $\frac{\tau}{2}$ in place of $\tau$). Suppose that $(p..q) \cap S \ne \emptyset$. Denote by $D$ the set of all pairs of positions $(i,j)$ such that $i,j  \in S$, $0 < j - i \le \tau /4$, and $s[i..i{+}\frac{\tau}{2}] = s[j..j{+}\frac{\tau}{2}]$. The set $S'$ is generated from $S$ by removing, for each pair $(i,j) \in D$, the positions $S \cap (i..j]$. Therefore, all positions of $S$ between $p$ and $q$ are covered by half-intervals $(i..j]$ with $(i,j) \in D$. Further, there is a subcover of $S \cap (p..q)$ consisting of interleaving half-intervals  $\{(i_t..j_t]\}_{t=1}^m$ from $D$, i.e., $i_1 < i_2 < \cdots < i_m$, $S \cap (p..q) \subseteq (i_1..j_m]$, and, for any $t \in [1..m)$, $i_{t+1} \le j_t$. Since, for each $(i,j) \in D$, the string $s[i..j{+}\frac{\tau}{2}]$ has a period at most $\tau / 4$, the period of the substring $s[i_1..j_m{+}\frac{\tau}{2}]$ is at most $\tau / 4$ by Lemma~\ref{lem:fine-wilf}. It is easy to see that $i_1 = p$. Hence, if $j_m + \frac{\tau}{2} \ge q$, then property~(c) for $S'$ is established: $s[p..q]$ has a period at most $\tau / 4$. Otherwise (if $j_m + \frac{\tau}{2} < q$), we have $S \cap (j_m..q) = \emptyset$ and, therefore, $s[j_m .. q]$ has a period at most $\tau / 8$, by property~(c) of $S$ (with $\frac{\tau}{2}$ in place of $\tau$). Then, $s[p..q]$ has a period at most $\tau / 4$ by Lemma~\ref{lem:fine-wilf} since it is covered by two overlapping substrings $s[p..j_m{+}\frac{\tau}{2}]$ and $s[j_m..q]$ with periods at most $\tau / 4$.

For the converse of property~(c), consider a substring $s[i..j]$ whose minimal period $\pi$ is at most $\tau / 4$. Let $q \in S \cap [i + \frac{3}{4}\tau .. j - \frac{3}{4}\tau]$. Denote $p = q - \pi$. Due to periodicity of $s[i..j]$, we obtain $s[q{-}\frac{\tau}{2} .. q{+}\frac{\tau}{2}] = s[p{-}\frac{\tau}{2} .. p{+}\frac{\tau}{2}]$. By property~(a) of $S$ (with $\frac{\tau}{2}$ in place of $\tau$), we have $p \in S$. Since $q \in (p .. p{+}\tau/4]$, the procedure generating $S'$ must have compared the strings $s[p..p{+}\frac{\tau}{2}]$ and $s[q..q{+}\frac{\tau}{2}]$ during the analysis of the position $p$ and, as a result, could not put $q$ into $S'$. Hence, $S' \cap [i + \frac{3}{4}\tau .. j - \frac{3}{4}\tau] = \emptyset$.
\end{proof}


\subsection{\boldmath Proof that $S^*$ obtained after the recompression is $\tau$-partitioning}
\label{appx:main-lemma}
\begin{lemma}
If $s[p..p{+}\frac{7}{8}\tau] = s[q..q{+}\frac{7}{8}\tau]$, for $p, q \in S'$, then $a_p = a_q$.\label{lem:r-letters-locality}
\end{lemma}
\begin{proof}
Since $S'$ is an almost $\frac{3}{4}\tau$-partitioning set, as stated in Lemma~\ref{lem:periodicity-gap}, it follows from property~(b) for $S'$ that the positions $p$ and $q$ have a common ``right context'' of length $\tau / 8$: more formally, for any $d \in [0..\tau/8]$, we have $s[p{+}d .. p{+}d{+}\tau/2] = s[q{+}d .. q{+}d{+}\tau/2]$, and $p + d \in S'$ iff $q + d \in S'$.
In order to produce $a_p$, our algorithm first consecutively computed numbers $\bar{t}, \bar{\bar{t}}, \bar{\bar{\bar{t}}}$ for $p$. Denote by $p_1, p_2, \ldots, p_m$ all positions of $S' \cap (p..p{+}\tau/2^5]$ in the increasing order. By construction, the number $\bar{t}$ depends on a ``right context'' of $p$ of length $\tau/2^5$: $\bar{t}$ is produced by comparing $s[p..p{+}\tau/2]$ to all strings $s[p_1..p_1{+}\tau/2], \ldots,$ $s[p_m..p_m{+}\tau/2]$ and, thus, $\bar{t}$ coincides with analogously computed numbers at any other positions $r$ such that, for all $d \in [0..\tau/2^5]$, $s[p{+}d .. p{+}d{+}\tau/2] = s[r{+}d .. r{+}d{+}\tau/2]$ and $p + d \in S'$ iff $r + d \in S'$. It remains to show that the numbers $\bar{\bar{t}}, \bar{\bar{\bar{t}}}, a_p$ in the same sense depend on ``right contexts'' of $p$ of lengths $2\tau/2^5, 3\tau/2^5, 4\tau/2^5 = \tau/8$, respectively. The proof is similar for all three cases. Consider, for instance, the number $\bar{\bar{\bar{t}}}$, assuming that the claim holds for $\bar{\bar{t}}$. We obtain $\bar{\bar{\bar{t}}}$ by comparing $\bar{\bar{t}}$ to numbers $\bar{\bar{t}}_1, \bar{\bar{t}}_2, \ldots, \bar{\bar{t}}_m$ computed for $p_1, p_2, \ldots, p_m$, respectively. By the assumption, the number $\bar{\bar{t}}_m$ corresponding to the rightmost of the positions, $p_m$, depends on a ``right context'' of $p_m$ of length $2\tau/2^5$. Therefore, since $p_m - p \le \tau/2^5$, we obtain the claimed dependency of $\bar{\bar{\bar{t}}}$ on a ``right context'' of $p$ with length $\tau/2^5 + 2\tau/2^5 = 3\tau/2^5$.
\end{proof}

\mainLemma*
\begin{proof}
Since $S^* \subseteq S'$, the converse of property~(c) is inherited from the set $S'$ that satisfies it by Lemma~\ref{lem:periodicity-gap}; we relax the condition $S^* \cap [i + \frac{3}{4}\tau .. j - \frac{3}{4}\tau] = \emptyset$ slightly, for aesthetical reasons.

For $h \in [\LBB..\lceil\log^{(4)} n\rceil]$, denote by $S_h$ the set of positions from $S'$ whose corresponding letters remained in $R$ after the algorithm has performed all recompression procedures for all $j > h$.
In particular, $S_{\lceil\log^{(4)} n\rceil} = S'$ and $S_{\LBB} = S^*$. Note that the size of each set $S_h$ is at most $2^{h+\BB} \cdot \frac{n}{\tau}$.

For property~(a) of $S^*$, consider $p$ and $q$ such that $s[p{-}\tau..p{+}\tau] = s[q{-}\tau..q{+}\tau]$. Let us show by induction that, for each $h \in [\LBB..\lceil\log^{(4)} n\rceil]$ and each $d$ such that $|d| \le \frac{1}{8}\tau - \frac{8}{2^{h+1}}\tau$, we have $p + d \in S_h$ iff $q + d \in S_h$. In particular, for $h = \LBB$, $p \in S^*$ iff $q \in S^*$, which is precisely the claim of property~(a). The base of the induction is $h = \lceil\log^{(4)} n\rceil$: since, as stated in Lemma~\ref{lem:periodicity-gap}, $S'$ is an almost $\frac{3}{4}\tau$-partitioning set, we have $p + d \in S'$ iff $q + d \in S'$, for any $d$ such that $|d| \le \frac{1}{8}\tau$. By Lemma~\ref{lem:r-letters-locality}, letters of $R$ corresponding to positions $p + d$ and $q + d$ such that $p + d \in S'$ and $q + d \in S'$ coincide provided $|d| \le \frac{1}{8}\tau$ since $s[p{+}d..p{+}d{+}\frac{7}{8}\tau] = s[q{+}d..q{+}d{+}\frac{7}{8}\tau]$.

Fix $h \in [\LB..\lceil\log^{(4)} n\rceil]$ and suppose, by the inductive hypothesis, that $p + d \in S_h$ iff $q + d \in S_h$, for any $d$ such that $|d| \le \frac{1}{8}\tau - \frac{8}{2^{h+1}}\tau$. We are to prove the inductive step: $p + d \in S_{h-1}$ iff $q + d \in S_{h-1}$, for any $d$ such that $|d| \le \frac{1}{8}\tau - \frac{8}{2^{h}}\tau$. Let $R$ be a string obtained by the algorithm after performing all recompression procedures for all $j > h$. Thus, there is a one-to-one correspondence between the positions $S_h$ and letters of the string $R$. Consider a (contiguous) sequence of letters $a_{i} a_{{i}+1} \cdots a_{m}$ of $R$ corresponding to all positions $p + d \in S_h$ such that $|d| \le \frac{1}{8}\tau - \frac{8}{2^{h+1}}\tau$, and, analogously, a sequence $a_{{i}'} a_{{i}'+1} \cdots a_{m'}$ corresponding to positions $q + d \in S_h$. By the inductive hypothesis and due to Lemma~\ref{lem:r-letters-locality}, the sequences coincide. The algorithm performs on the string $R$ at most three recompression reductions removing some of the letters until the positions of $S_h$ corresponding to the remained letters constitute the set $S_{h-1}$. Denote by $r$ and $r'$ the positions of $S_h$ corresponding to the letters $a_m$ and $a_{m'}$, respectively. A discrepancy in the processing of the sequences by the first iteration of recompression may occur only in their last letters: for instance, the letter $a_m$ will be removed whereas $a_{m'}$ will be retained (see an example in Figure~\ref{fig:discrepancies}). Let us analyze this particular case (other cases are similar). This may happen only if $a_m \in \acute{\Sigma}$ and $a_m$ is followed in $R$ by a letter $a \in \grave{\Sigma}$ and the distance between $r$ and the position of $S_h$ corresponding to $a$ is at most $\tau / 2^h$; at the same time, either $a_{m'}$ ($=a_m$) is followed by a different letter $b \in \acute{\Sigma}$ ($\ne a$) or the distance between $r'$ and the position of $S_h$ corresponding to this following letter is larger than $\tau / 2^h$. We therefore deduce that the distance from $r$ to $p + \frac{1}{8}\tau - \frac{8}{2^{h+1}}\tau$ is less than $\tau / 2^h$ (for otherwise the letter $a$ following $a_m$ must be a part of the sequence $a_{i} a_{{i}+1} \cdots a_{m}$), which implies that $r - p > \frac{1}{8}\tau - \frac{8}{2^{h+1}}\tau - \frac{1}{2^h} \tau = \frac{1}{8}\tau - \frac{5}{2^h}\tau$.

Thus, we have shown that two sequences resulting from $a_{i} a_{{i}+1} \cdots a_{m}$ and $a_{{i}'} a_{{i}'+1} \cdots a_{m'}$ after the recompression coincide in all letters whose corresponding positions from $S_h$ are at a distance at most $\frac{1}{8}\tau - \frac{5}{2^h}\tau$ from $p$ and $q$, respectively. Exactly the same argument can be applied for the second recompression. A discrepancy in the processing of the resulting sequences by the second iteration of recompression (if any) may again occur only in last letters of the sequences: for instance, the last letter $c$ of the first sequence (the letter $a_{m-1}$ in the example of Figure~\ref{fig:discrepancies}) will be retained whereas the corresponding letter $c$ of the second sequence (the letter $a_{m'-1}$ in the example)
will be removed.
Denote by $r''$ the position of $S_h$ corresponding to the removed letter $c$. By analogy to the argument used for the first iteration, we deduce that the distance between $r''$ and $r'$ is at most $\tau / 2^h$. Observe that $r - p = r' - q$. Therefore, $r'' - q \ge r' - q - \frac{1}{2^h}\tau > \frac{1}{8}\tau - \frac{5}{2^h}\tau - \frac{1}{2^h}\tau = \frac{1}{8}\tau - \frac{6}{2^h}\tau$.
Thus, two sequences resulting from the second recompression coincide in all letters whose corresponding positions from $S_h$ are at a distance at most $\frac{1}{8}\tau - \frac{6}{2^h}\tau$ from $p$ and $q$, respectively. Analogously, we deduce that two sequences resulting from the third recompression (if any) coincide in all letters whose corresponding positions from $S_h$ are at a distance at most $\frac{1}{8}\tau - \frac{6}{2^h}\tau - \frac{1}{2^h}\tau = \frac{1}{8}\tau - \frac{7}{2^h}\tau$ from $p$ and $q$, respectively. This proves the inductive claim since $\frac{1}{8}\tau - \frac{7}{2^h}\tau \ge \frac{1}{8}\tau - \frac{8}{2^h}\tau$ and the set $S_{h-1}$ consists of all positions from $S_h$ that correspond to the letters of $R$ remained after the (at most) three recompressions.

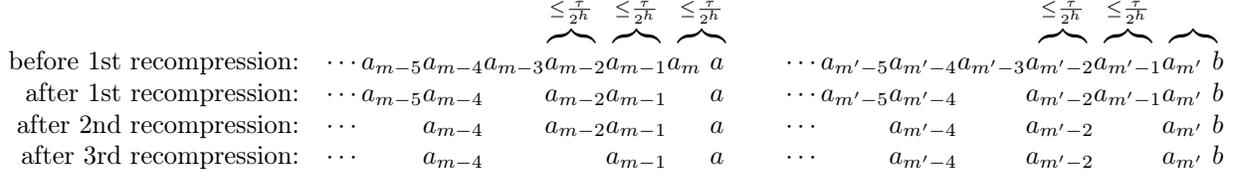
\begin{figure}[!hbt]
\[
\begin{array}{rl@{}cl}
\text{before 1st rec.} &\hspace{-2mm} \cdots a_{m-5}           a_{m-4} {a_{m-3}}         a_{m-2}   \rn{m1}{a_{m-1}} a_{m}           ~a &  & \cdots a_{m'-5}           a_{m'-4} {a_{m'-3}}         a_{m'-2} \rn{m2}{a_{m'-1}} a_{m'} ~b\\
\text{after 1st rec.}  &\hspace{-2mm} \cdots {a_{m-5}}         a_{m-4} \phantom{a_{m-3}} a_{m-2}           a_{m-1} \phantom{a_{m}} ~a &  & \cdots {a_{m'-5}}         a_{m'-4} \phantom{a_{m'-3}} a_{m'-2} {a_{m'-1}}         a_{m'} ~b\\
\text{after 2nd rec.}  &\hspace{-2mm} \cdots \phantom{a_{m-5}} a_{m-4} \phantom{a_{m-3}} a_{m-2}           a_{m-1} \phantom{a_{m}} ~a &  & \cdots \phantom{a_{m'-5}} a_{m'-4} \phantom{a_{m'-3}} a_{m'-2} \phantom{a_{m'-1}} a_{m'} ~b\\
\text{after 3rd rec.}  &\hspace{-2mm} \cdots \phantom{a_{m-5}} a_{m-4} \phantom{a_{m-3}} \phantom{a_{m-2}} a_{m-1} \phantom{a_{m}} ~a &  & \cdots \phantom{a_{m'-5}} a_{m'-4} \phantom{a_{m'-3}} a_{m'-2} \phantom{a_{m'-1}} a_{m'} ~b
\end{array}
\]
\begin{tikzpicture}[overlay, remember picture]
\node[anchor=south] at (m1) {$\overbrace{\phantom{aa}}^{{\le}\frac{\tau}{2^h}}~\overbrace{\phantom{aa}}^{{\le}\frac{\tau}{2^h}} ~\overbrace{\phantom{aa}}^{{\le}\frac{\tau}{2^h}}$};
\node[anchor=south] at (m2) {$\overbrace{\phantom{aa}}^{{\le}\frac{\tau}{2^h}} ~\overbrace{\phantom{aa}}^{{\le}\frac{\tau}{2^h}} ~\overbrace{\phantom{aa}}^{~}$};
\end{tikzpicture}
\caption{A schematic example of three iterations of recompression (``rec.'' on the figure) on equal sequences $a_{{i}} a_{{i}+1} \cdots a_{m}$ and $a_{{i}'} a_{{i}'+1} \cdots a_{m'}$ ($m - {i} = m' -{i}'$ and $a_{{i} + t} = a_{{i}' + t}$, for any $t \in [0..m{-}{i}]$). The overbraces designate restrictions on distances between positions of $S_h$ corresponding to letters.}\label{fig:discrepancies}
\end{figure}

For property~(b) of $S^*$, consider $p, q \in S^*$ such that $s[p..p{+}\ell] = s[q..q{+}\ell]$, for some $\ell \ge 0$. Clearly, only the case $\ell > \tau$ is interesting. Denote $\tilde{\ell} = \ell - \tau$. It suffices to prove that, for each $h \in [\LBB..\lceil\log^{(4)} n\rceil]$ and each $d$ such that $0 \le d < \tilde{\ell} + \frac{1}{8}\tau - \frac{8}{2^{h+1}}\tau$, we have $p + d \in S_h$ iff $q + d \in S_h$. In particular, for $h = \LBB$, it is precisely the claim of property~(b): for any $d \in [0..\tilde{\ell}) = [0..\ell{-}\tau)$, $p + d \in S^*$ iff $q + d \in S^*$.
The proof is essentially by the same inductive argument as for property~(a). The base of the induction, $h = \lceil\log^{(4)} n\rceil$, follows from Lemma~\ref{lem:periodicity-gap} where it is stated that $S'$ ($= S_{\lceil\log^{(4)} n\rceil}$) is an almost $\frac{3}{4}\tau$-partitioning set. By Lemma~\ref{lem:r-letters-locality}, the letters of $R$ corresponding to positions $p + d$ and $q + d$ such that $p + d \in S'$ and $q + d \in S'$ coincide provided $0 \le d < \tilde{\ell} + \frac{1}{8}\tau$ since $s[p{+}d..p{+}d{+}\frac{7}{8}\tau] = s[q{+}d..q{+}d{+}\frac{7}{8}\tau]$. The proof of the inductive step is very similar to the proof for property~(a) above; we therefore only briefly sketch it without diving into details.

Fix $h \in [\LB..\lceil\log^{(4)} n\rceil]$ and suppose, by the inductive hypothesis, that $p + d \in S_h$ iff $q + d \in S_h$, for any $d$ such that $0 \le d < \tilde{\ell} + \frac{1}{8}\tau - \frac{8}{2^{h+1}}\tau$. Let $R$ be a string produced by the algorithm after performing all recompression procedures for all $j > h$. There is a one-to-one correspondence between $S_h$ and letters of $R$. Consider a (contiguous) sequence of letters $a_i a_{i+1} \cdots a_{m}$ of $R$ corresponding to all positions $p + d \in S_h$ such that $0 \le d < \tilde{\ell} + \frac{1}{8}\tau - \frac{8}{2^{h+1}}\tau$, and, analogously, a sequence $a_{i'} a_{i'+1} \cdots a_{m'}$ corresponding to positions $q + d \in S_h$. By the inductive hypothesis and due to Lemma~\ref{lem:r-letters-locality}, the sequences coincide. The algorithm performs on $R$ at most three recompressions removing some letters until the positions of $S_h$ corresponding to the remained letters constitute the set $S_{h-1}$. Discrepancies occurring in the two sequences after the recompressions may affect only last positions that are close---at a distance at most $3\tau / 2^h$---to the ``right borders'', $p + \tilde{\ell} + \frac{1}{8}\tau - \frac{8}{2^{h+1}}\tau$ and $q + \tilde{\ell} + \frac{1}{8}\tau - \frac{8}{2^{h+1}}\tau$,  of the sequences, respectively.
Therefore, all positions from $S_h \cap [p..\infty)$ and $S_h \cap [q..\infty)$ that were at a distance at least $3\tau / 2^h$ to the left of the ``right borders'' are ``synchronized'', i.e., such a position $p + d \in S_h$ is removed from $S_h$ iff $q + d \in S_h$ is removed too. The ``synchronized'' positions are exactly $p + d \in S_h$ and $q + d \in S_h$ such that $0 \le d < \tilde{\ell} + \frac{1}{8}\tau - \frac{8}{2^{h+1}}\tau - \frac{3}{2^h}\tau = \tilde{\ell} + \frac{1}{8}\tau - \frac{7}{2^{h}}\tau$, so that $p + d \in S_{h-1}$ iff $q + d \in S_{h-1}$. Since $\frac{1}{8}\tau - \frac{7}{2^{h}}\tau \ge \frac{1}{8}\tau - \frac{8}{2^{h}}\tau$ and the set $S_{h-1}$ is formed by remained positions of $R$ after the (at most) three recompressions, this proves the inductive step.

For property~(c) of $S^*$, consider $p, q \in S^*$ such that $q - p > \tau$ and $S^* \cap (p..q) = \emptyset$. By construction, a recompression procedure performed on a string $R$ may delete a letter $R[i]$ only if the distance from the position $r$ corresponding to $R[i]$ to the positions $r'$ and $r''$ of $S'$ corresponding to the letters $R[i{-}1]$ and $R[i{+}1]$ is at most $\tau / 2^5$. Further, if $R[i]$ is removed, then neither $R[i{-}1]$ nor $R[i{+}1]$ can be removed in the same iteration of recompression. The distance between $r'$ and $r''$ is at most $\tau / 2^4$. Therefore, it is impossible that there was a position of $S'$ from $(p..q)$ that got removed. Thus, $S' \cap (p..q) = \emptyset$. Since $S'$ is a $\tau$-partitioning set, we deduce that the substring $s[p..q]$ has a period at most $\tau / 4$ by property~(c) of $S'$.
\end{proof}

\section{\boldmath Small $\tau$}
\label{appx:small-tau}

Assume that $\tau < (\log^{(3)} n)^4$. If $\tau \ge 2^5 \log^{(3)} n$, we perform the procedure of Sections~\ref{sec:vishkin-process}--\ref{sec:time-improvement} generating a \mbox{$\frac{\tau}{2}$-partitioning} set $S$ of size $\Oh(\frac{n}{\tau}\log^{(3)} n)$. If $\tau < 2^5 \log^{(3)} n$, we put $S = [0..n)$, which is a \mbox{$\frac{\tau}{2}$-partitioning} set of size $\Oh(\frac{n}{\tau}\log^{(3)} n)$ in this case. All this is done in $\Oh(n)$ time and $\Oh(\frac{n}{\tau})$ space (provided the set $S$ is not stored explicitly), which are the time and space bounds we aim for. The problem is in the remaining stages of the algorithm. Specifically, it is in the following ``slow'' parts:
\begin{enumerate}
\item[(i)] the generation of the subset $S' \subseteq S$ requires $\Oh(n + |S| \log^{(3)} n)$ time;
\item[(ii)] the construction of the string $R$ takes $\Oh(n + |S'|(\log^{(3)} n)^3)$ time;
\item[(iii)] the initialization of the arrays $M_i$, for all $i \in [0..|S'|)$, requires $\Oh(|S'|\log^{(3)} n)$ time;
\item[(iv)] each update of the arrays $M_i$ before shrinking the string $R$ takes  $\Oh(|R|\log^{(4)} n)$ time.
\end{enumerate}

If we optimize all these four bottlenecks to run, respectively, in $\Oh(n)$, $\Oh(n)$, $\Oh(|S'|)$, and $\Oh(|R|)$ time, then the running time of the whole algorithm will be $\Oh(n)$. We do all four optimizations by the four russians' trick~\cite{ArlazarovDinicKronrodFaradzev}. The idea of the trick is that if one has to perform complicated queries on chunks of $c$ bits, then instead of computing the queries explicitly each time, we can precalculate a table of size $2^c$ with answers for every possible chunk. What are the ``chunks'' and the ``queries'' in our case? The easiest is to analyze~(iv) to give an example. Let us describe how one can update the arrays $M_i$ in $\Oh(|R|)$ time before shrinking $R$.

\subparagraph{Part (iv).} Suppose that we have marked some letters of $R$ for removal using a bit array $A[0..|R|{-}1]$ of length $R$: $R[i]$ is marked iff $A[i] = 1$. Every array $M_i[0..\lceil\log^{(4)} n\rceil]$ occupies $\Oh((\log^{(4)} n)^2)$ bits (each entry is of size $\Oh(\log^{(4)} n)$ bits). The na{\"i}ve updating procedure considers every entry $M_i[j]$ and counts which of the letters $R[i{+}1], R[i{+}2], \ldots, R[i{+}M_i[j]]$ were not marked for removal, assigning the result to $M_i[j]$. Recall that each number $M_i[j]$, for $j \in [0..\lceil\log^{(4)} n\rceil]$, is at most $\Oh(\log^{(3)} n)$. Therefore, the new value for $M_i[j]$ is determined by the old value $M_i[j]$ and the subarray $A[i..i{+}\Oh(\log^{(3)} n)]$, for an appropriate constant under the big-O. By storing the bit array $A$ in a sequence of $\Oh(1 + |R| / \log n)$ machine words of size $\Theta(\log n)$ bits, we can retrieve and pack the subarray $A[i..i{+}\Oh(\log^{(3)} n)]$ in $\Oh(1)$ time into one machine word and can concatenate the subarray to the bit representation of the whole $M_i$. The bit representation of $M_i$ takes $\Oh((\log^{(4)} n)^2)$ and, thus, the resulting chunk after the concatenation occupies $\Oh((\log^{(4)} n)^2) + \Oh(\log^{(3)} n) = \Oh(\log^{(3)} n)$ bits. We view this chunk storing the concatenated bit representations of $M_i[0..\lceil\log^{(4)} n\rceil]$ and $A[i..i{+}\Oh(\log^{(3)} n)]$ as an integer number $x$ with $\Oh(\log^{(3)} n)$ bits. It is clear that the chunk determines the state of the updated array $M_i$. Therefore, we can in advance before the start of the whole algorithm consider all possible valid chunks that encode in the same way arrays $M[0..\lceil\log^{(4)} n\rceil]$ with entries of size $\Oh(\log^{(4)} n)$ bits concatenated with bit arrays of length $\Oh(\log^{(3)} n)$ and we can precalculate the updated arrays $M$ in a table $B$ so that the updated array $M_i$ is already stored in the entry $B[x]$. Thus, we simply read $B[x]$, which contains a bit representation for the updated array $M_i$ occupying $\Oh((\log^{(4)} n)^2)$ bits, and we rewrite $M_i$ with the content of $B[x]$. All is done $\Oh(1)$ time since all the bit representations take only $\Oh(1)$ machine words. The size of the table $B$ is $\Oh(2^{\Oh(\log^{(3)} n)} \cdot (\log^{(4)} n)^2) = o(\log n)$ bits and, hence, all precalculations can be performed in $o(n)$ time with $\Oh(1)$ space (the space is measured in $\Oh(\log n)$-bit machine words).

\subparagraph{Part (i).}
Let us describe how the set $S'$ is produced in $\Oh(n)$ time from the set $S$, which is fed to our algorithm from left to right in an online fashion. The na{\"i}ve procedure, for each position $p \in S$, collects all positions from the set $S \cap (p..p{+}\tau/4]$ and compares $s[p..p{+}\tau/2]$ to each of the substrings $s[q..q{+}\tau/2]$ with $q \in S \cap (p..p{+}\tau/4]$; if $s[p..p{+}\tau/2] = s[q..q{+}\tau/2]$, then all the positions $S \cap (p..q]$ are marked for removal and, thus, will not be reported into $S'$. The idea of a faster solution is that all the comparisons in this algorithm, for the given $p$, are performed inside a very short substring $s[p..p{+}\tau/4+\tau/2]$ (recall that $\tau < (\log^{(3)} n)^4$); we therefore can pack all this string into one small chunk that fits into one machine word specifying which of its positions are from $S$ and, then, we perform the marking for removal in $\Oh(1)$ time using a precomputed answer. The issue with this idea is that the alphabet can be quite large so that the packing is impossible.

A solution is to ``reduce'' the alphabet for the substring $s[p..p{+}\tau/4+\tau/2]$ by sorting all letters and substituting them by their ordering numbers. The sorting can be performed in linear time using fusion trees~\cite{FredmanWillard,PatrascuThorup2}. More precisely, during the left-to-right processing of $S$, we consecutively consider (overlapping) substrings $s[i\tau..(i{+}2)\tau]$, for $i \in [0..n/\tau-2]$. Using the fusion tree, we can sort all letters of a given substring $s[i\tau..(i{+}2)\tau]$ in $\Oh(\tau)$ time assigning to them their ordering numbers, i.e., reducing the alphabet to $[0..2\tau]$ (see~\cite{FredmanWillard,PatrascuThorup2}; note that we have $\Omega(n / \tau) = \Omega(n / \log n)$ space). Denote by $\hat{s}_i$ the string $s[i\tau..(i{+}2)\tau]$ with letters substituted by the ordering numbers (so that all new letters are from $[0..2\tau]$). The string $\hat{s}$ occupies $\Oh(\tau\log\tau) = o((\log^{(3)} n)^5)$ bits and, thus, can be packed into one machine word. Now in order to process $p \in S$, we first check whether the substring $s[p..p{+}\tau/4{+}\tau/2]$ is contained in the last preprocessed substring $\hat{s}_i$. If not, we set $i = \lfloor p / \tau\rfloor$ and preprocess $s[i\tau..(i{+}2)\tau]$ in $\Oh(\tau)$ time generating $\hat{s}_i$. Note that in this way we never preprocess the same substring $s[i\tau..(i{+}2)\tau]$ twice since the positions $p \in S$ are fed to our algorithm from left to right. Next, using standard bit operations on the machine word containing $\hat{s}_i$, we retrieve the substring $s[p..p{+}\tau/4{+}\tau/2]$ from the string $s[i\tau..(i{+}2)\tau]$ encoded in $\hat{s}_i$ in which the alphabet is ``reduced'', i.e., all letters are substituted with numbers from $[0..2\tau]$. The retrieved substring occupies $\Oh(\tau\log\tau)$ bits and is stored in one machine word. We concatenate to the bit representation of this substring a bit array $a$ of length $\tau/4$ that indicates which of the positions $i + 1, i + 2, \ldots, i + \tau/4$ belong to $S$: for $h \in (0..\tau/4]$, we have $i + h \in S$ iff $a[h - 1] = 1$. The bit array $a$ is easy to maintain in one machine word during the execution of the algorithm using the bit shift operations. Thus, we obtain a chunk of $\Oh(\tau\log\tau) + \tau/4 = o((\log^{(3)} n)^5)$ bits that encodes the concatenated bit representations of the substring $s[p..p{+}\tau/4{+}\tau/2]$, with a ``reduced'' alphabet, and of the bit array $a$, indicating all positions from $S \cap (p..p{+}\tau/4]$. We view this chunk as an integer number $x$ with $o((\log^{(3)} n)^5)$ bits.

Clearly, the chunk $x$ determines which of the positions from $S \cap (p..p{+}\tau/4]$ should be marked for removal during the processing of $p$. Therefore, we can in advance before the start of the whole algorithm consider all possible valid chunks that encode in the same way strings $t[0..\tau/2+\tau/4]$ over the alphabet $[0..2\tau]$ concatenated with bit arrays of length $\tau/4$ and we can precalculate which of the positions from $S \cap (p..p{+}\tau/4]$ will be marked for removal in a table $C$ so that the entry $C[x]$ stores a bit array $c$ of length $\tau/4$ that ``masks'' positions for removal: for $h\in (0..\tau/4]$, $c[h - 1] = 0$ iff the position $q \in S \cap (p..p{+}\tau/4]$ such that $|S \cap (p..q]| = h$ must be marked for removal because there exists $q' \in S \cap [q..p{+}\tau/4]$ such that $s[p..p{+}\tau/2] = s[q'..q'{+}\tau/2]$. Thus, we simply read $C[x]$ and apply the masking array stored in $C[x]$ to mark for removal some positions from $S$ after $p$.

Thus, the processing of $p\in S$ is done in $\Oh(1)$ time since all the bit representations take only $\Oh(1)$ machine words. The total time for processing all substrings $s[i\tau..(i{+}2)\tau]$, for $i \in [0..n/\tau-2]$, is $\Oh(\frac{n}{\tau}\tau) = \Oh(n)$. The size of the table $C$ is $\Oh(2^{o((\log^{(3)} n)^5)} \cdot \tau) = o(\log n)$ bits and, hence, all precalculations can be performed in $o(n)$ time with $\Oh(1)$ space. As in the original algorithm from Section~\ref{sec:recompression}, the computation of $S'$ is executed in an online manner reporting its positions from left to right without storing them (only few are stored for internal purposes).

\subparagraph{Part (ii).}
Let us describe now how the string $R$ can be computed from the set $S'$, which is fed to our algorithm from left to right in an online fashion. The idea is very similar to what was done for Part~(i) to produce the set $S'$ from $S$. The algorithm starts with an empty string $R$ and considers all $p \in S'$ from left to right generating, for each $p \in S'$, a letter $a_p$ that is then appended to the end of the string $R$. As is evident from Lemma~\ref{lem:r-letters-locality}, the construction of $a_p$ requires only the substring $s[p..p{+}\frac{7}{8}\tau]$ and the positions $S' \cap (p..p{+}\frac{7}{8}\tau]$. At first glance, the same trick can be applied as in Part~(i): we consecutively consider substrings $s[i\tau..(i{+}2)\tau]$, for $i \in [0..n/\tau-2]$, reducing their alphabets to $[0..2\tau]$ and encoding them into chunks $\hat{s}_i$ of $\Oh(\tau\log\tau)$ bits; when a position $p \in S'$ arrives, we retrieve the substring $s[p..p{+}\frac{7}{8}\tau]$ with a reduced alphabet from the chunk $\hat{s}_i$ with $i = \lfloor p / \tau\rfloor$, constructing $\hat{s}_i$ from scratch if it was not built previously, and then, we concatenate to the chunk a bit array of length $\Oh(\tau)$ indicating which of the positions from $(p..p{+}\frac{7}{8}\tau]$ belong to $S'$. Unfortunately, this scheme does not work since the alphabet reductions loose some essential information required to construct the letters $a_p$. Precalculations therefore are more involved.

Denote by $p_1, p_2, \ldots, p_m$ the set of all positions from $S' \cap (p..p{+}\tau/2^5]$ in increasing order. Recall that the procedure of Section~\ref{sec:recompression} first generates for $p$ a tuple of numbers $\langle w'_1, w'_2, \ldots, w'_\ell\rangle$: for $j \in [1..\ell]$, $w'_j = \vbit(t, t_j)$ if $j \le m$, and $w'_j = \infty$ otherwise, where $t = \sum_{i=0}^{\tau/2} s[p{+}i] 2^{wi}$ and $t_j = \sum_{i=0}^{\tau/2} s[p_j{+}i] 2^{wi}$, for $j \in [1..m]$. The numbers $t$ and $t_j$ simply represent the substrings $s[p..p{+}\tau/2]$ and $s[p_j..p_j{+}\tau/2]$. In order to compute $\vbit(t, t_j)$, it suffices to find the length $L$ of the longest common prefix of $s[p..p{+}\tau/2]$ and $s[p_j..p_j{+}\tau/2]$ and, then, compute the lowest bit at which the numbers $s[p{+}L]$ and $s[p_j{+}L]$ differ and which of these numbers has 0 and 1 in this differing bit. While the length $L$ can be computed on the substring $s[p..p{+}\frac{7}{8}\tau]$ with a reduced alphabet, the information about the bits at which the numbers $s[p{+}L]$ and $s[p_j{+}L]$ differ is lost after the alphabet reduction. However, it turns out that this information can be stored too in small space without the need to preserve the numbers $s[p{+}L]$ and $s[p_j{+}L]$ themselves.

Consider the procedure reducing the alphabet for a substring $s[i\tau..(i{+}2)\tau]$ with $i \in [0..n/\tau-2]$. The procedure sorts all letters assigning to them ordering numbers from $[0..2\tau]$. Denote by $b_0, b_1, \ldots, b_{k-1}$ all distinct letters of $s[i\tau..(i{+}2)\tau]$ in increasing order before the reduction. Each letter $b_i$ is an $\Oh(\log n)$-bit number that is mapped by the reduction to its index $i$. Each number $b_i$ can be represented as a bit string $\bar{b}_i$ of length $\Oh(\log n)$ in which the bits are written from the lowest to the highest. We can construct a compacted trie on the strings $\bar{b}_i$. Since $k \le 2\tau + 1$, the compacted trie can be stored in $\Oh(\tau\log\log n)$ bits: every edge in the trie stores the length of the bit string written on it (which takes $\Oh(\log\log n)$ bits), each internal node contains pointers to its children (taking $\Oh(\log\tau)$ bits since the number of nodes is $\Oh(\tau)$), and each leaf stores the index $i$ of the corresponding number $b_i$ (taking $\Oh(\log\tau)$ bits). The compacted trie can be built in $\Oh(k)$ time using a fusion tree~\cite{FredmanWillard}; in fact, the fusion tree on the numbers $b_0, b_1, \ldots, b_{k-1}$ implicitly constructs precisely this trie (see also~\cite{ChanLarsenPatrascu,GrossiOrlandiRamanRao} where this is emphasized more explicitly). Our idea is that we do not have to store the numbers $b_0, b_0, \ldots, b_{k-1}$ in addition to the trie in order to find the lowest bit at which two numbers $b_i$ and $b_{i'}$ differ: the bit corresponds precisely to the position in the compacted trie at which the corresponding strings $\bar{b}_i$ and $\bar{b}_{i'}$ diverge.

By analogy to the solution for Part~(i), when processing $p \in S'$, we first retrieve the substring $s[p..p{+}\frac{7}{8}\tau]$ with a reduced alphabet from the string $s[i\tau..(i{+}2)\tau]$ with $i = \lfloor p / \tau\rfloor$ encoded in $\hat{s}_i$. Then, we concatenate to the bit representation of this substring a bit array $a$ of length $\lfloor\frac{7}{8}\tau\rfloor$ that indicates which of the positions $i + 1, i + 2, \ldots, i + \lfloor\frac{7}{8}\tau\rfloor$ belong to $S'$: for $h \in (0..\frac{7}{8}\tau]$, we have $i + h \in S'$ iff $a[h - 1] = 1$. The bit array $a$ is easy to maintain in one machine word during the execution of the algorithm using the bit shift operation. Next, we concatenate to the resulting bit chunk a bit representation of the compacted trie on the bit strings $\bar{b}_0, \bar{b}_1, \ldots, \bar{b}_{k-1}$ (without storing the letters $b_0, b_0, \ldots, b_{k-1}$ themselves), which adds more $\Oh(\tau\log\log n)$ bits. Thus, we obtain a chunk of $\Oh(\tau\log\tau) + \frac{7}{8}\tau + \Oh(\tau\log\log n) = \Oh((\log^{(3)} n)^4\log\log n)$ bits that encodes the concatenated bit representations of the substring $s[p..p{+}\frac{7}{8}\tau]$ with a reduced alphabet, of the bit array $a$ indicating all positions from $S' \cap (p..p{+}\frac{7}{8}\tau]$, and of the compacted trie. We view this chunk as an integer number $x$ with $\Oh((\log^{(3)} n)^4\log\log n)$ bits.

The chunk $x$ determines the letter $a_p$. Indeed, in order to compute the letter, we first have to compute the numbers $w'_j = \vbit(t, t_j)$: this can be done by first computing the length $L$ of the longest common prefix of the corresponding substrings of $s[p..p{+}\frac{7}{8}\tau]$ at positions $p$ and $p_j$ and, then, by finding the lowest differing bit of the numbers $s[p{+}L]$ and $s[p_j{+}L]$, which can be performed using the compacted trie; similar computations should be executed for other positions from $S' \cap (p..p{+}\frac{7}{8}]$ but all of them involve only substrings of the string $s[p..p{+}\frac{7}{8}\tau]$, due to Lemma~\ref{lem:r-letters-locality}. All further computations can be executed as described in Section~\ref{sec:recompression}. Instead of performing all this from scratch, we can in advance consider all possible valid chunks that encode in the same way strings $t[0..\frac{7}{8}\tau]$ over the alphabet $[0..2\tau]$ concatenated with bit arrays of length $\frac{7}{8}\tau$ and with compacted tries on $\Oh(\tau)$ bit strings of length $\Oh(\log n)$ (without storing the $\Oh(\tau)$ strings themselves); for each such chunk, we precalculate the resulting letter in a table $D$ so that the entry $D[x]$ stores $a_p$. Thus, we simply read $D[x]$ and append it to $R$.

The processing of $p\in S$ is done in $\Oh(1)$ time since all the bit representations take only $\Oh(1)$ machine words. The total time for processing all substrings $s[i\tau..(i{+}2)\tau]$, for $i \in [0..n/\tau-2]$, is $\Oh(\frac{n}{\tau}\tau) = \Oh(n)$. The size of the table $D$ is $\Oh(2^{\Oh((\log^{(3)} n)^4\log\log n)} \cdot (\log^{(3)} n)^2)$ bits, which fits into $\Oh(n/\tau)$ space provided $\tau < (\log^{(3)} n)^4$ as can be easily seen since $\log(n/\tau) = \Theta(\log n)$ whereas the logarithm of the space for $D$ is $o(\log\log n)$. Hence, all precalculations can be performed in $o(n)$ time within $\Oh(n / \tau)$ space.

\subparagraph{Part (iii).}
It remains to describe how all the arrays $M_i[0..\lceil\log^{(4)} n\rceil]$ can be initialized in $\Oh(|S'|)$ time. Given a letter $R[i]$ and a position $p \in S'$ corresponding to it, recall that $M_i[j]$, for $j \in [0..\lceil\log^{(4)} n\rceil]$, is equal to the size of the set $S' \cap (p..p{+}\tau/2^j]$. We receive positions $p$ of $S'$ from left to right and, using the bit shift operation, maintain along the way a bit array $a$ of length $\tau$ that is stored in a machine word $x$ of size $\Oh(\log n)$ bits such that $a$ indicates which of the positions of $(p..p{+}\tau]$ are from $S'$, i.e., for $h \in (0..\tau]$, we have $p + h \in S'$ iff $a[h - 1] = 1$. Obviously, the array $a$ determines the content of $M_i$. The array $M_i$ occupies $\Oh((\log^{(4)} n)^2)$ bits and, therefore, its bit representation can be stored into one machine word. We hence can precompute a table $F$ of size $\Oh(2^{\tau}\cdot (\log^{(4)} n)^2) = o(\log n)$ bits such that $F[x]$ stores the bit representation of the array $M_i$. The content of $F[x]$ is then copied in $\Oh(1)$ time in place of $M_i$. The table $F$ can be straightforwardly precalculated in $o(n)$ time and $\Oh(1)$ space.

\fi
\end{document}